\documentclass[11pt]{article}

\def\math#1{$#1$}

\def\frac#1#2{{#1\over #2}}

\DeclareSymbolFont{AMSb}{U}{msb}{m}{n}
\DeclareMathSymbol{\N}{\mathbin}{AMSb}{"4E}
\DeclareMathSymbol{\Z}{\mathbin}{AMSb}{"5A}
\DeclareMathSymbol{\R}{\mathbin}{AMSb}{"52}
\DeclareMathSymbol{\Q}{\mathbin}{AMSb}{"51}
\DeclareMathSymbol{\I}{\mathbin}{AMSb}{"49}
\DeclareMathSymbol{\C}{\mathbin}{AMSb}{"43}

\newcommand{\blue}[1]{{\color[named]{Blue} #1}}
\newcommand{\green}[1]{{\color[named]{Green} #1}}

\def\dotfil{\leaders\hbox to 1.5mm{.}\hfill}
\newcounter{rmnum}
\def\RN#1{\setcounter{rmnum}{#1}\uppercase\expandafter{\romannumeral\value{rmnum}}}
\def\rn#1{\setcounter{rmnum}{#1}\expandafter{\romannumeral\value{rmnum}}}

\usepackage{color,graphics,enumerate,wrapfig,setspace,fullpage}
\usepackage{array,graphics,color,epsfig}
\usepackage[noend]{algorithmic}
\usepackage{amsmath}
\usepackage{amsthm}

\newtheorem{lemma}{Lemma}
\newtheorem{theorem}{Theorem}

\title{
Extracting Hidden Groups and their 
Structure from Streaming Interaction Data 
}

\author{
Mark K. Goldberg \\
{\sf goldberg@cs.rpi.edu}
\and
Mykola Hayvanovych \\
{\sf hayvam@cs.rpi.edu}
\and 
Malik Magdon-Ismail \\
{\sf magdon@cs.rpi.edu}
\and
William A. Wallace \\
{\sf wallaw@rpi.edu}
}

\date{
Rensselaer Polytechnic Institute,\\
110 8th Street, Troy, NY 12180, USA.\\
{\sf \{goldberg,hayvam,magdon\}@cs.rpi.edu; wallaw@rpi.edu}.\\
\today
}

\begin{document}
\sloppy
\maketitle

\begin{abstract}
When actors in a social network interact, it usually means they 
have some general goal towards which they are collaborating. This could 
be a research collaboration in a company or a foursome planning a
golf game. We call such groups \emph{planning groups}. In many 
social contexts, it might be possible to observe the 
\emph{dyadic interactions} between actors, even if the actors do not 
explicitly declare what groups they belong too. When groups are not
explicitly declared, we call them \emph{hidden groups}. Our particular focus is
hidden planning groups. By virtue of their need to further their
goal, the actors within such 
groups must interact in a manner
which differentiates their communications
from random background communications. 
In such a case, one can infer (from these interactions)
 the composition  and structure of the hidden planning groups.
We formulate the problem of 
hidden group discovery from streaming interaction data, and 
we propose efficient algorithms for identifying 
the hidden group structures by isolating the hidden group's 
non-random, planning-related, 
communications from the random background communications.
We validate our algorithms on real data
(the Enron email corpus 
and Blog communication data).
Analysis of the results reveals that our algorithms 
extract meaningful hidden group structures. 
\end{abstract}

\section{Introduction}
\label{section:intro}

Communication networks (telephone, email, Internet chatroom, etc.) 
facilitate rapid information exchange among millions of users around 
the world,
providing the ideal environment for groups
to plan their activity undetected:
their communications are embedded (hidden) within the myriad of
unrelated communications.
A group may communicate in a structured way while
not being forthright about its existence.
However, when the group must exchange 
communications to plan some activity, their \emph{need}
 to communicate
usually imposes some structure on their communications.
We develop  statistical and algorithmic approaches for 
discovering such hidden groups that \emph{plan} an activity.
Hidden group members may have non-planning related communications,
be malicious (e.g. a terrorist
group) 
or benign (e.g. a golf foursome).
We liberally use ``hidden group''
for all such groups involved in planning, even  though they
may not intentionally be hiding their communications.

The tragic events of September 11, 2001
 underline the need for a tool  which 
aides in the discovery of 
hidden  groups during their {\it planning} 
stage, before implementation.
One approach to discovering such groups is using
{\it correlations} among the group member communications.
The {\it communication graph} of the society
 is defined by its 
actors (nodes) and communications (edges).
We do not use communication content,
even though it
can be informative through some natural 
language processing, because
such analysis 
is time consuming and intractable for large datasets.
We use only the time-stamp, sender and recipient ID of a message.

Our approach of discovering hidden groups is based on the observation that
the pattern of communications exhibited by a group
pursuing a common objective
is different from that of a randomly selected set of actors:
any group, even one which tries to hide itself,
must communicate regularly to plan.
One possible instance of such correlated communication
is the occurrence of a
{\it repeated communication pattern.}
Temporal correlation emerges as the members of a group 
need to systematically exchange messages to 
plan their future
activity. 
This correlation among the group communications
will exist throughout the 
planning stage, which may be some extended period of
time. If the planning occurs over a long enough period, this temporal
correlation will stand out
against a random background of communications and hence can be
detected.

\section{Streaming Hidden Groups}

\begin{figure}[t]
\centering
{\small
\begin{tabular}{p{2in}p{0.62in}}
\multicolumn{1}{m{3in}}{\footnotesize
\begin{tabular}{cll}
00&\blue{{\bf A}\math{\rightarrow}{\bf C}}
&\blue{Golf tomorrow? Tell everyone.}\\
05&\blue{{\bf C}\math{\rightarrow}{\bf F}}
&\blue{Alice mentioned golf tomorrow.}\\
06&\blue{{\bf A}\math{\rightarrow}{\bf B}}
&\blue{Hey, golf tomorrow? Spread the word}\\
12&\green{{\bf A}\math{\rightarrow}{\bf B}}
&\green{Tee time: 8am; Place: Pinehurst.}\\
13&\blue{{\bf F}\math{\rightarrow}{\bf G}}
&\blue{Hey guys, golf tomorrow .}\\
13&\blue{{\bf F}\math{\rightarrow}{\bf H}}
&\blue{Hey guys, golf tomorrow .}\\
15&\green{{\bf A}\math{\rightarrow}{\bf C}}
&\green{Tee time: 8am; Place: Pinehurst.}\\
20&\blue{{\bf B}\math{\rightarrow}{\bf D}}
&\blue{We're playing golf tomorrow.}\\
20&\blue{{\bf B}\math{\rightarrow}{\bf E}}
&\blue{We're playing golf tomorrow.}\\
22&\green{{\bf C}\math{\rightarrow}{\bf F}}
&\green{Tee time: 8am; Place: Pinehurst.}\\
25&\green{{\bf B}\math{\rightarrow}{\bf D}}
&\green{Tee time: 8am; Place: Pinehurst.}\\
25&\green{{\bf B}\math{\rightarrow}{\bf E}}
&\green{Tee time 8am, Pinehurst.}\\
31&\green{{\bf F}\math{\rightarrow}{\bf G}}
&\green{Tee time 8am, Pinehurst.}\\
31&\green{{\bf F}\math{\rightarrow}{\bf H}}
&\green{Tee off 8am,Pinehurst.}
\end{tabular}
}
&
\multicolumn{1}{m{1in}}{\footnotesize
\begin{tabular}{cl}
00&{{\bf A}\math{\rightarrow}{\bf C}}\\
05&{{\bf C}\math{\rightarrow}{\bf F}}\\
06&{{\bf A}\math{\rightarrow}{\bf B}}\\
12&{{\bf A}\math{\rightarrow}{\bf B}}\\
13&{{\bf F}\math{\rightarrow}{\bf G}}\\
13&{{\bf F}\math{\rightarrow}{\bf H}}\\
15&{{\bf A}\math{\rightarrow}{\bf C}}\\
20&{{\bf B}\math{\rightarrow}{\bf D}}\\
20&{{\bf B}\math{\rightarrow}{\bf E}}\\
22&{{\bf C}\math{\rightarrow}{\bf F}}\\
25&{{\bf B}\math{\rightarrow}{\bf D}}\\
25&{{\bf B}\math{\rightarrow}{\bf E}}\\
31&{{\bf F}\math{\rightarrow}{\bf G}}\\
31&{{\bf F}\math{\rightarrow}{\bf H}}\\
\end{tabular}
}
\\
\centerline{(a)}&
\centerline{(b)}
\end{tabular}
}
\caption{
Streaming hidden group with two waves of planning (a).  
Streaming group without message content -- only time, sender id and 
receiver id are available (b).}
\label{fig:stream}
\end{figure}

Unlike in the cyclic hidden group setting \cite{hidden2004} where  
all of the hidden group members communicate within 
some characteristic time period, and do so repeatedly over a consecutive
sequence of time periods. 
A \emph{streaming hidden group} doesn't obey such strict 
requirements for its communication pattern.
Hidden groups don't necessarily
display a fixed time-cycle, during which
all members of group members exchange messages, but whenever a step 
in the planning needs to occur, some hidden group member initiates
a communication, which percolates through the hidden group.
The hidden group problem may still be formulated as one of
finding repeated (possibly overlapping)
communication patterns.
An example of a streaming hidden group is illustrated in 
Fig.~\ref{fig:stream}(a) with the same group planning
golf game.
Given the message content, it is easy to identify
two ``waves'' of communication. The first wave (in darker font)
establishes 
the golf game; the second wave (in lighter font)
finalizes the game details. 
Based on this
data, it is not hard to identify the 
group and conclude that the ``organizational structure''
of the group is represented in Fig.~\ref{fig:stream_graph} 
to the right (each actor is
represented by their first initial). The challenge, once again,
 is to deduce this
same information from the communication stream \emph{without} the message 
contents Fig.~\ref{fig:stream}(b).
Two features that distinguish the stream from the
cycle model are:

(\rn{1})
communication waves may overlap, as in Fig.~\ref{fig:stream}(a);

(\rn{2})
waves may have different durations, some considerably longer than others.

The first feature may result in bursty waves of intense communication (many
overlapping waves) followed by periods of silence. Such a type of communication 
dynamics is hard to detect in the cycle model, since all the (overlapping)
waves of communication may fall in one cycle. The second can be quantified by
a propagation delay function which specifies how much time may elapse between
a hidden group member receiving the message and forwarding it to the next
member; sometimes the propagation delays may be large, and sometimes small.
One would typically expect 
\begin{wrapfigure}{r}{1.3in}
\vspace*{-0.1in}
\resizebox{1.2in}{!}{\includegraphics*{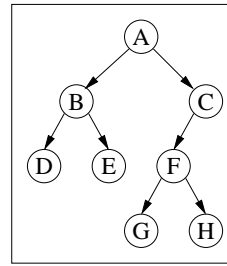}}
\caption{\small Group structure in Fig. \ref{fig:stream}}
\label{fig:stream_graph}
\end{wrapfigure}
that such a streaming model would be appropriate for hidden groups with
some organizational structure as illustrated in the tree
in Fig.~\ref{fig:stream_graph}. We present algorithms which 
discover the streaming hidden group and its organizational structure 
without the use of message content.

We use the notion of communication frequency in order to distinguish nonrandom 
behavior. Thus, if a group of actors communicates unusually often using
the same chain of communication, i.e. the structure of their communications
persists through time, then we consider this group to be statistically significant
and indicative of a hidden group. We present algorithms to detect small frequent
tree-like structures, and build hidden structures starting from the small ones.

\section{Our Contributions}
We present
efficient
algorithms 
which
not only discover the streaming hidden group, but also its organizational
structure \emph{without the use of message content}.
We use the notion of communication frequency
in order to distinguish non-random behavior.
Thus, if a group of actors communicates unusually often using
the same chain of communication,
i.e. the structure of their communications persists through time, 
then we consider this 
group to be statistically anomalous.
We present algorithms to detect small frequent tree-like 
structures, and build hidden structures starting from the small ones.
We also propose an approach that uses new cluster matching algorithms 
together with a sliding window technique to track and observe 
the evolution of hidden groups over time. 
We also present a general query algorithm which can determine
if a given hidden group (represented as a tree) occurs frequently
in the communication stream. Additionally we propose efficient algorithms 
to obtain the frequency of general trees and to enumerate all statistically 
significant general trees of a specified size and frequency. 
Such algorithms are used in conjunction with the heuristic algorithms and 
similarity measure techniques to 
verify that a discovered tree-like structure actually occurs frequently in 
the data. 
We validate our algorithms on the Enron email corpus, as well as the 
Blog communication data. 

\emph{Paper Organization.}
First we consider related work,
followed by the methodologies for the streaming hidden groups and tree mining 
in Section~\ref{sec:stream}. 
Next we present similarity measure methods in 
Section~\ref{sec:Measuring Similarity between Sets of Overlapping Clusters}. 
We present experiments on real world data and validation results 
in Section~\ref{sec:enron-data} followed by the summary and conclusions 
in Section~\ref{sec:conclusion}. 

\section{Related Work}
\label{sec:relatedwork}

Identifying structure in networks has been studied 
extensively in the context of clustering and partitioning (see for example 
\cite{bach2002,cluster2005a,cluster2005b,clauset2005,cluster2005bb,capocci2004,flake2002,
girvan2002,hendrickson1995,kannan2004,karypis1998b,kernighan1970,dimacs2003,newman2003}). 
These approaches focus on static, non-planning, hidden groups.
In \cite{hidden2003} Hidden Markov models are the basis for discovering
planning hidden groups. The underlying
methodology is based on random graphs \cite{bollobas3,janson1}
and some of the results on cyclic hidden groups were
presented in \cite{hidden2004}.
In our work we incorporate some of the prevailing social science theories, such as 
homophily (\cite{monge2002}), by incorporating
group structure.
More models of societal evolution and simulation
can be found in
\cite{carley1,carley2,malik28,sanil1,malik30,Palla2007,222956,networkmech5} which deal with
dynamic models for social network infrastructure, rather than the dynamics of
the actual communication behavior. 

Our work is novel because we detect hidden groups
by only analyzing
communication intensities (and not message content).
The study of streaming hidden groups was initiated in 
\cite{hidden2006}, which contains some preliminary results. 
We extend these results and present a general query algorithm 
which can find if a given hidden group (represented as a tree) occurs frequently
in the communication stream, which we extended to algorithm 
to obtain the frequency of general trees and to enumerate all statistically 
significant general trees of a specified size and frequency. 
Such algorithms are used in conjunction with the heuristic algorithms 
and similarity measures to verify that a discovered tree-like 
structure actually occurs frequently in the data. 

Erickson,
\cite{secretsoc}, was one of the first to  study
secret societies. His focus was on general
communication structure.
Since the
September 11, 2001 terrorist plot,
discovering hidden groups became a topic of intense research.  
For example it was understood that
Mohammed Atta
was central to the planning, but that a large percent 
of the network would need to be removed to render it inoperable
\cite{sixdegrees,uncloaking}. 
Krebs,
\cite{uncloaking} identified the network as sparse,
which renders it hard to discover through clustering 
in the 
traditional sense (finding dense subsets).
Our work on temporal correlation would address exactly such a situation.
It has also been observed that terrorist  group structure may be changing
\cite{netwars},
and our methods are based on connectivity which is immune to this trend. 
We assume that message authorship is known,
which may not be true, Abbasi and Chen propose 
techniques to address this issue, \cite{abbasi2005}.

\section{Problem Statement}
\label{sec:ProblemStatement}

A hidden group communication structure can be represented by a 
directed graph. Each vertex is an actor and every edge shows 
the direction of the communication. For example a hierarchical organization 
structure could be represented by a directed tree.
The graph in Figure \ref{fig:commun_struct} to the right
is an example of a communication structure, in which 
actor $A$ ``simultaneously'' sends messages to $B$ and $C$;
then, after receiving the message from $A$, $B$ sends messages to $C$ and $D$; 
$C$ sends a message to $D$ after receiving the messages from $A$ and $B$. 
Every graph has two basic types of communication structures: \emph{chains} and 
\emph{siblings}. A \emph{chain} is a path of length at least
3, and a \emph{sibling} is a tree with 
a root and  two or more children, but no other nodes. 
Of particular interest are chains and sibling trees with three nodes,
which we denote \emph{triples}.
For example, the chains and sibling trees 
of size three (triples) 
in the communication structure above are:
$A \rightarrow B \rightarrow D$; 
$A \rightarrow B \rightarrow C$; 
$A \rightarrow C \rightarrow D$;
$B \rightarrow C \rightarrow D$; 
$A \rightarrow (B,C)$;
and,
$B\rightarrow (C,D)$.
We suppose that a hidden group employs a communication structure that 
can be represented by a directed graph as above. If the hidden group is 
hierarchical, the communication graph will be a tree.
The task is to discover such a group and its structure based solely on 
the communication data.

If  a communication structure appears in the data many times, then it is
likely to be non-random, and hence represent a hidden group. To discover
hidden groups, we will discover the communication structures that appear 
many times. We thus need to define what it means for a communication
structure to ``appear''.
\begin{wrapfigure}{r}{1.4in}
\vspace*{-0.4in}
  \begin{center}
   \resizebox{1.2in}{1.2in} {\input{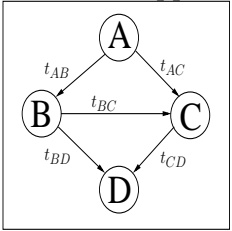}}
  \end{center}
  \caption{\small Target group.}
  \label{fig:commun_struct}
\end{wrapfigure}
Specifically, we consider 
chain and 
sibling triples (trees of size three). For a 
chain $A \rightarrow B \rightarrow C$ to appear, 
there must be communication $A \rightarrow B$ 
at time $t_{AB}$ and a communication $B \rightarrow C$ 
at time $t_{BC}$ such that 
$(t_{BC} - t_{AB}) \in [\tau_{min}, \tau_{max}]$. This intuitively 
represents the notion of causality, where $A \rightarrow B$ ``causes''
 $B \rightarrow C$ within some time 
interval specified by $[\tau_{min}$, $\tau_{max}]$. 
A similar requirement holds for the sibling triple $A \rightarrow B,C$;
the sibling triple 
appears if there exists $t_{AB}$ and $t_{AC}$ such that 
 $(t_{AB} - t_{AC}) \in [-\delta$ $\delta]$. 
This constraint represents the 
notion of $A$ sending messages ``simultaneously''
 to $B$ and $C$ within a small time interval 
of each other, as specified by \math{\delta}.
For an entire graph (such as the one above) to appear, every chain 
and sibling triple in the graph must appear 
using a single set of times. For example, in the graph example
above, there must exist 
a set of times,
\math{\{{t_{AB},t_{AC},t_{BC},t_{BD},t_{CD}}\}}, 
which 
satisfies all the six chain and sibling constraints: 
$(t_{BD} - t_{AB}) \in [\tau_{min}, \tau_{max}]$, 
$(t_{BC} - t_{AB}) \in [\tau_{min}, \tau_{max}]$, 
$(t_{CD} - t_{AC}) \in [\tau_{min}, \tau_{max}]$, 
$(t_{CD} - t_{BC}) \in [\tau_{min}, \tau_{max}]$, 
$(t_{AB} - t_{AC}) \in [- \delta, \delta]$ and 
$(t_{BD} - t_{BC}) \in [- \delta, \delta]$. 
A graph appears multiple times if there are disjoint sets of times 
each of which is an appearance 
of the graph. 
A set of times \emph{satisfies}
 a 
graph if all chain and sibling constraints are satisfied by the set of 
times. 
The number of times a graph 
appears is the maximum number of disjoint sets of times
that can be found, where each set 
satisfies the graph. Causality requires that multiple occurrences of a graph should
monotonically increase in time. Specifically,
if $t_{AB}$ ``causes'' \math{t_{BC}} and 
\math{t'_{AB}} ``causes'' \math{t'_{BC}} with \math{t'_{AB}>t_{AB}},
then it should be that \math{t'_{BC}>t_{BC}}.
In general,
if we have two disjoint occurrences (sets of times)
\math{\{t_1,t_2,\ldots\}} and \math{\{s_1,s_2,\ldots\}}
with \math{s_1>t_1}, then it should be that \math{s_i>t_i} for all \math{i}.
A communication structure which is frequent
enough becomes statistically significant when its frequency 
exceeds the expected frequency of such a structure from
the random background communications. 
The goal is to find all statistically significant communication
structures, which is formally stated in the following algorithmic
problem statement.

{\bf Input:}  
A communication data stream and parameters: $\delta$, $\tau_{min}$, $\tau_{max}$, 
$h$, $\kappa$.

{\bf Output:}
 All communication structures
 of size $\geq h$, which appear at least $\kappa$ times, where 
the appearance is defined with respect to 
$\delta$, $\tau_{min}$, $\tau_{max}$. 
\smallskip

Assuming we can solve this algorithmic task, the statistical task is to determine 
\math{h} and $\kappa$ 
to ensure that all the output communication structures reliably correspond to 
non-random ``hidden groups''. We first consider small trees, 
specifically chain and sibling triples.
We then develop a heuristic to build up larger hidden groups 
from clusters of triples. Additionally we mine all of the frequent 
directed acyclic graphs and propose new ways of measuring the similarity 
between sets of overlapping sets. We 
obtain evolving hidden groups by using a sliding window in conjunction with 
the proposed similarity measures to determine the rate of evolution.

\label{sec:stream}
\section{Algorithms for Chain and Sibling Trees}
\label{sec:Algorithms for Chain and Sibling Trees}

We will start by introducing a technique to find chain and sibling triples, 
i.e. trees of type $A \rightarrow B \rightarrow C$ (chain) 
and trees of type $A \rightarrow (B,C)$ (sibling). To accomplish this, we 
will enumerate all the triples and count the number 
of times each triple occurs. Enumeration can be done by brute force, i.e.
 considering each possible triple 
in the stream of communications. We have developed a general algorithm for 
counting the number of occurrences of chains of length $\ell$, 
and siblings of width $k$. These algorithms proceed by posing the problem 
as a multi-dimensional matching problem, which in the case of tipples 
becomes a two-dimensional matching problem. Generally multi-dimensional 
matching is hard to solve, but in our case the causality 
constraint imposes an ordering 
on the matching which allows us to construct a linear time algorithm. 
Finally we will introduce a heuristic to build larger 
graphs from statistically significant triples using overlapping clustering 
techniques \cite{cluster2005a}.

\subsection{Computing the Frequency of a Triple}
Consider the triple $A \rightarrow B \rightarrow C$ and the associated 
time lists $L_1 = \{ t_1 \leq t_2 \leq \ldots \leq t_n \}$ and 
$L_2 = \{ s_1 \leq s_2 \leq \cdots \leq s_m \}$, where $t_i$ are the times 
when $A$ sent to $B$ and $s_i$ the times when $B$ sent to $C$. 
An occurrence of the triple $A \rightarrow B \rightarrow C$ is a pair of 
times ($t_i$,$s_i$) such that $(s_i - t_i) \in [\tau_{min}$ 
$\tau_{max}]$. Thus, we would like to find the maximum number of such pairs 
which satisfy the causality constraint. It turns out that 
the causality constraint does not affect the size of the maximum matching, 
however it is an intuitive constraint in our context. 

We now define a slightly more general maximum matching problem: for a 
pair $(t_i, s_i)$ let $f(t_i, s_i)$ denote the score of the pair. 

Let $M$ be a matching $\{ (t_{i_1},s_{i_1}),(t_{i_2},s_{i_2}) 
\ldots (t_{i_k},s_{i_k}) \}$ of size $k$. 
We define the score of $M$ as 

\[
Score(M) = \sum_{j=1}^k f(t_{i_j},s_{i_j}).
\]
The maximum matching problem is to find a matching with a maximum score. 
The function $f(t,s)$ captures how likely a message from 
$B \rightarrow C$ at time $s$ was ``caused'' by a message from $A 
\rightarrow B$ at time $t$. In our case we are using a hard threshold function 

\[
\ f(t,s)=f(t-s)= \left\{
\begin{array}{ll}
1 & \mbox{if } t-s $ $ \in $ $ [\tau_{min}, \tau_{max}] , \\
0 & \mbox{otherwise} .\\
\end{array} \right. 
\]
The matching problem for sibling triples is identical with the choice 

\[
\ f(t,s)=f(t-s)= \left\{
\begin{array}{ll}
1 & \mbox{if } t-s $ $ \in $ $ [-\delta, \delta] , \\
0 & \mbox{otherwise} .\\
\end{array} \right. 
\]
We can generalize to chains of arbitrary length and siblings of arbitrary 
width as follows. Consider time lists 
$L_1$, $L_2$, $\ldots$ ,$L_{\ell -1}$ corresponding to the chain $A_1 
\rightarrow A_2 \rightarrow A_3 \rightarrow \cdots \rightarrow A_{\ell}$, where $L_i$ contains the 
sorted times of communications $A_i \rightarrow A_{i+1}$. 
An occurrence of this chain is now an $\ell -1$ dimensional matching 
$\{t_1,t_2, \ldots , t_{\ell -1}\}$ satisfying the 
constraint ($t_{i+1} - t_i$) $\in$ [$\tau_{min}$ $\tau_{max}$] 
$\forall$ $i=1$,$\cdots$,$\ell -2$. 

The sibling of width $k$ breaks down into two cases: ordered siblings 
which obey constraints similar to the chain constraints, 
and unordered siblings. Consider the sibling tree $A_0 \rightarrow A_1,A_2, 
\cdots A_k$ with corresponding time lists 
$L_1$, $L_2$, $\ldots$ ,$L_k$, where $L_i$ contains the times of 
communications $A_0 \rightarrow A_i$. An occurrence is a matching 
$\{t_1,t_2, \ldots , t_k\}$. In the ordered case the constraints 
are $(t_{i+1} - t_i) \in [-\delta$ $\delta]$. 
This represents $A_0$ sending communications ``simultaneously'' 
to its recipients in the order $A_1, \ldots ,A_k$. The unordered sibling tree 
obeys the stricter constraint $(t_i - t_j)$ $\in$ $[-(k-1)\delta, 
(k-1)\delta]$, $\forall$ $i,j$ pairs, $i \neq j$. This stricter constraint 
represents $A_0$ sending communications to its 
recipients ``simultaneously'' without any particular order.

Both problems can be solved with a greedy algorithm. 
The detailed algorithms for arbitrary chains and siblings are given in 
Figure \ref{fig:algs}(a). 
Here we sketch the algorithm for triples. Given two 
time lists $L_1$=$\{ t_1, t_2, \ldots, t_n \}$ and 
$L_2$=$\{ s_1, s_2, \ldots, s_m \}$ the idea is to find the first 
valid match $(t_{i_1}, s_{i_1})$, which is the first pair 
of times that obey the constraint 
$(s_{i_1} - t_{i_1}) \in  [\tau_{min}$ $\tau_{max}]$, then 
recursively find the maximum matching on the 
remaining sub lists $L_1'=\{ t_{i_1+1}, \ldots, t_n \}$ 
and $L_2'=\{ s_{i_1+1}, \ldots, s_m \}$.

The case of general chains and ordered sibling trees is similar.
The first 
valid match is defined similarly. 
Every pair of entries $t_{L_i}$ $\in$ $L_i$ and $t_{L_{i+1}}$ 
$\in$ $L_{i+1}$ in the maximum matching must obey 
the constraint $(t_{L_{i+1}} - t_{L_i}) \in [\tau_{min}$ $\tau_{max}]$. 
To find the first valid match, 
we begin with the match consisting of the first time in all lists. 
Denote these times $t_{L_1}, t_{L_2}, \ldots, t_{L_{\ell}}$. 
If this match is valid (all consecutive pairs satisfy the constraint) 
then we are done. Otherwise consider the 
first consecutive pair to violate this constraint. Suppose it is 
$(t_{L_i}, t_{L_{i+1}})$; so either 
$(t_{L_{i+1}} - t_{L_i}) >  \tau_{max}$ or $(t_{L_{i+1}} - t_{L_i}) <  
\tau_{min}$. If $(t_{L_{i+1}} - t_{L_i}) >  \tau_{max}$ 
($t_{L_i}$ is too small), we advance $t_{L_i}$ to the next entry in the 
time list $L_i$; otherwise 
$(t_{L_{i+1}} - t_{L_i}) <  \tau_{min}$ ($t_{L_{i+1}}$ is too small) 
and we advance $t_{L_{i+1}}$ to the next entry in the 
time list $L_{i+1}$. This entire process is repeated until a valid 
first match is found. An efficient implementation of this 
algorithm is given in Figure \ref{fig:algs}. The algorithm for 
unordered siblings follows a similar logic.

\begin {figure}
\begin{center}
\hspace*{-0.1in}
\begin {tabular}{c|c}		
\begin{minipage}{2.8in} 
\begin{algorithmic}[1]
\STATE {\bf Algorithm Chain}
\WHILE {$P_k \leq \|L_k\|-1, \forall k$}
\IF {$(t_j - t_i) < \tau_{min}$}
\STATE {$P_j \gets P_j+1$}
\ELSIF {$(t_j - t_i) \in [\tau_{min}, \tau_{max}]$}
\IF {$j = n$}
\STATE {$(P_1,\ldots,P_n)$ is the next match}
\STATE {$P_k \gets P_k+1, \forall k$}
\STATE {$i \gets 0; j \gets 1$}
\ELSE
\STATE {$i \gets j; j \gets j+1$}
\ENDIF
\ELSE
\STATE {$P_i \gets P_i+1$}
\STATE {$j \gets i; i \gets i-1$}
\ENDIF
\ENDWHILE
\end{algorithmic}
\end{minipage}
&
\begin{minipage}{2.8in}
\begin{algorithmic}[1]
\STATE {\bf Algorithm Sibling}
\WHILE {$P_k \leq \|L_k\|-1, \forall k$}
\IF {$(t_j - t_i) < -(k-1)\delta$}
\STATE {$P_j \gets P_j+1$}
\ELSIF {$(t_j - t_i) > (k-1)\delta, \forall i<j$}
\STATE {$P_i \gets P_i+1$}
\STATE {$j \gets i+1$}
\ELSE
\IF {$j = n$}
\STATE {$(P_1,\ldots,P_n)$ is the next match}
\STATE {$P_k \gets P_k+1, \forall k$}
\STATE {$i \gets 0; j \gets 1$}
\ELSE
\STATE {$j \gets j+1$}
\ENDIF
\ENDIF
\ENDWHILE
\end{algorithmic}
\end{minipage}
\\
\\
{(a)}&
{(b)}
\end{tabular}
\caption{Maximum matching algorithm for chains and ordered siblings (a); 
Maximum matching algorithm for unordered siblings (b). 
In the algorithms above, we initialize $i = 0; j = 1$ ($i,j$ are 
time list positions), and $P_1, \ldots, P_n = 0$ ($P_k$ 
is an index within $L_k$ ). Let $t_i = L_i[P_i]$ and $t_j = L_j[P_j]$.}
\label{fig:algs}
\end{center}
\end{figure}

The next theorem gives the correctness of the algorithms.

\begin{theorem} 
Algorithm-Chain and Algorithm-Sibling find maximum matchings.
\end{theorem}

\begin{proof}
By induction. Given a set of time lists $L=(L_1, L_2, \ldots, L_n)$ 
our algorithm produces a matching 
$M=(m_1, m_2, \ldots, m_k)$, where each matching $m_i$ is a sequence of 
$n$ times from each of the $n$ time 
lists $m_i=(t_1^{i}, t_2^{i}, \ldots, t_n^{i})$. 
Let $M^{*}=(m^{*}_1, m^{*}_2, \ldots, m^{*}_{k^{*}})$ be a maximum 
matching of size $k^*$. We prove that \math{k=k^*}
by induction on 
$k^*$. The next lemma follows directly from the construction of the 
Algorithms.

\begin{lemma}
\label{lemma1}
If there is a valid matching our algorithm will find one.
\end{lemma}

\begin{lemma}
\label{lemma2}
Algorithm-Chain and Algorithm-Sibling find an earliest valid matching. 
Let the first valid matching found by either algorithm 
be $m_1 = (t_1, t_2, \ldots, t_n)$, then for any other valid 
matching $m' = (s_1, s_2, \ldots, s_n)$ $t_i \leq s_i$ 
$\forall$ $i = 1,\cdots,n$.
\end{lemma}

\begin{proof}
Proof by contradiction. 
Assume that in $m_1$ and $m'$ there exists a corresponding pair of times $s < t$ 
and let $s_i$, $t_i$ be the first such pair. 
Since $m_1$ and $m'$ are valid matchings, then $s_i$ and $t_i$ obey the constraints: 
$ \tau_{min} \leq (t_{i+1} - t_i) \leq \tau_{max}$, 
$\tau_{min} \leq (t_i - t_{i-1}) \leq \tau_{max}$ 
and $\tau_{min} \leq (s_{i+1} - s_i) \leq \tau_{max}$, 
$\tau_{min} \leq (s_{i} - s_{i-1}) \leq \tau_{max}$.

Since $s_i < t_i$, then $\tau_{min} < (t_{i+1} - s_i)$ and $\tau_{max} > (s_i - t_{i-1})$. 
Also because $s_{i-1} \geq t_{i-1}$, we get that $ \tau_{min} \leq (s_i - t_{i-1})$ and 
since $(s_{i+1} - s_i) \leq \tau_{max}$, then $(min(t_{i+1},s_{i+1}) - s_i) \leq \tau_{max}$ 
as well. But if $s_i$ satisfies the above conditions, 
then $m_1$ would not be the first valid matching, 
because the first matching $m_f$ would contain 
$m_f = (t_1, t_2, \ldots , t_{i-1}, s_i, min(t_{i+1}, s_{i+1}), min(t_{i+2}, s_{i+2}), \ldots , min(t_{n}, s_{n}))$. 

Let us show this by induction on the 
number of pairs $p$ of the type $min(t_{i+j}, s_{i+j})$, 
where $s_i < t_i$ and $j \geq 1$. 

If $p = 1$, then $j = 1$, and since $\tau_{min} \leq (s_{i+1} - s_i) \leq \tau_{max}$ and 
$\tau_{min} < (t_{i+1} - s_i)$, then $\tau_{min} < (min(t_{i+1}, s_{i+1}) - s_i) \leq \tau_{max}$ 
as well, and therefore satisfies the matching constraints.

Let the matching constraints be satisfied up to $p = m$, such that in the matching 
$m^* = (t_1, t_2, \ldots , t_{i-1}, s_i, min(t_{i+1}, s_{i+1}), \ldots,  min(t_{i+m}, s_{i+m}), \ldots , min(t_{n}, s_{n}))$ 
the sequence of elements of $m^*$ up to $min(t_{i+m}, s_{i+m})$ satisfy the matching constraints. 
Then we can show that $min(t_{i+m+1}, s_{i+m+1})$ is also a part of the matching. Since 
$m_1$ and $m'$ are both valid matchings, then $ \tau_{min} \leq (t_{i+m+1} - t_{i+m}) \leq \tau_{max}$ 
and $\tau_{min} \leq (s_{i+m+1} - s_{i+m}) \leq \tau_{max}$, from which we get that 
$\tau_{min} \leq (min(t_{i+m+1}, s_{i+m+1}) - min(t_{i+m}, s_{i+m})) \leq \tau_{max}$. 
Therefore, $min(t_{i+m+1}, s_{i+m+1})$ is also a part of the matching.

Thus, we get a contradiction since $m_f$ would be an earlier matching 
if there exists a pair of times $s_i < t_i$. Therefore, 
Algorithm-Chain and Algorithm-Sibling find an earliest valid matching.
\end{proof}

If $k^* = 0$, then $k = 0$ as well. If $k^* = 1$, then 
there exists a valid matching and by Lemma \ref{lemma1} 
our algorithm will find it.

Suppose that for all sets of time lists for which $k^* = M$, 
the algorithm finds matchings of size $k^*$.
Now consider a set of time lists $L=(L_1, L_2, \ldots, L_n)$ 
for which an optimal algorithm produces a maximum matching of size 
$k^* = M+1$ and consider the first matching in this list 
(remember that by the causality constraint, the matchings can be ordered). 
Our algorithm constructs the earliest matching and then 
recursively processes the remaining lists. By Lemma
 \ref{lemma2}, our first matching 
is not later than optimal's first matching, so the partial 
lists remaining after our first matching contain the partial lists 
after optimal's first matching. This means that the optimal 
matching for our partial lists must be $M$. By the induction hypothesis our 
algorithm finds a matching of size $M$ on these partial lists 
for a total matching of size $M+1$.
\end{proof}

For a given set of time lists $L=(L_1, L_2, \ldots, L_n)$ as 
input, where each $L_i$ has a respective size $d_i$, define the 
total size of the data as $\|D\|=\sum_{i=1}^n{d_i}$.

\begin{theorem} 
Algorithm-Chain runs in $O(\|D\|)$ time.
\end{theorem}


\begin {theorem} 
Algorithm-Sibling runs in $O(n\cdot\|D\|)$ time.
\end{theorem}


\subsection{Finding all Triples}

Assume the data are stored in a vector. Each component in the 
vector corresponds
to a sender id and stores a balanced 
search tree of receiver lists (indexed by a receiver id). 
And let $S$ be the whole set of distinct senders. The algorithm for 
finding chain triples considers sender 
id $s$ and its list of receivers $\{ r_1, r_2, \cdots, r_d \}$. Then for each 
such receiver $r_i$ that is also a sender, let $\{ \rho_1, \rho_2, 
\cdots, \rho_f \}$ be the receivers to which $r_i$ sent messages. 
All chains beginning with $s$ are of the form $s \rightarrow r_i \rightarrow 
\rho_j$. This way we can more efficiently 
enumerate the triples (since we ignore triples which do not occur). 
For each sender $s$ we count the frequency of 
each triple $s \rightarrow r_i \rightarrow \rho_j$. 

\begin {theorem} 
Algorithm to find all triple frequencies takes $O(\|D\| + n \cdot \|D\|)$ time.
\end{theorem}


\subsection{General Scoring Functions for $2D$-Matching}

One can observe that for our $2D$-matching we are using a so called 
``Step Function'', which returns $1$ for values between 
$[ \tau_{min}$,  $\tau_{max} ]$, and gives $0$ otherwise. 
Such a function represents the probability 
delay density which is the distribution of 
the time it takes to propagate a message once 
it is received.

\begin{figure}[h]
  \begin{center}
  \begin{tabular}{c c}
   \resizebox{7cm}{!} {\includegraphics[scale=1]{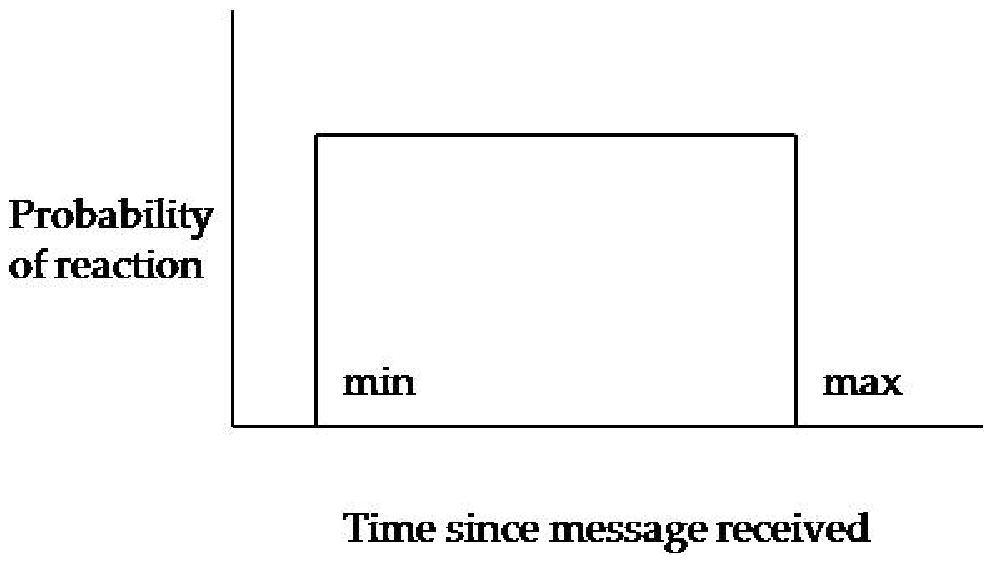}}
   &
   \resizebox{7cm}{!} {\includegraphics[scale=1]{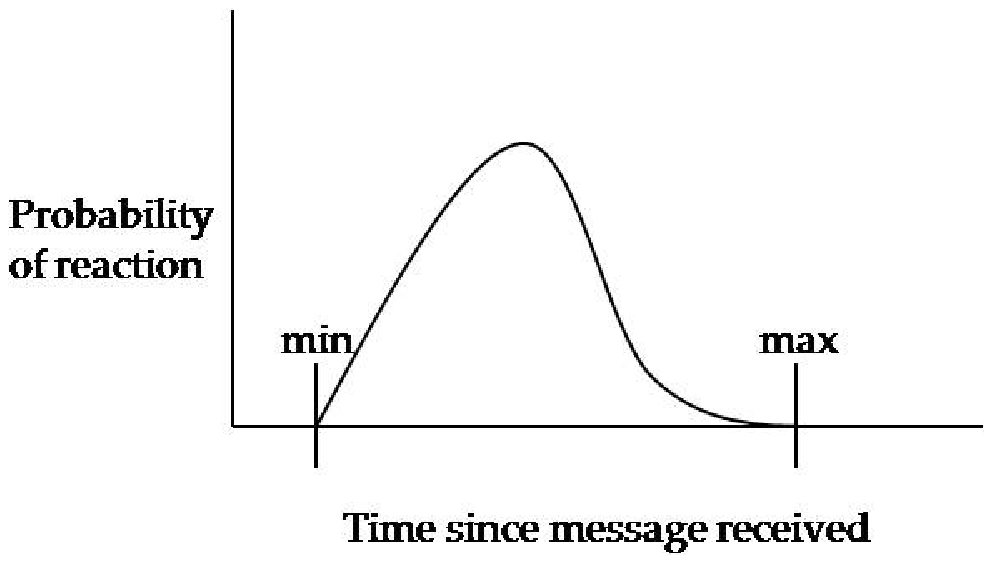}} 
   \\
   \end{tabular}
  \end{center}
  \caption{Step function on the left and a General Response Functions for 2D Matching on the right}
  \label{fig:funct}
\end{figure}

Here we extend our matching algorithm to be able to use any general 
propagation delay density function, see Figure \ref{fig:funct}. 

Usage of these various functions may uncover 
some additional information about the streaming groups and their 
structure which the ``Step Function'' missed. 

Unfortunately, the matching problem with an arbitrary function, 
unlike in the case with the ``Step Function'' which can be solved in linear time, 
cannot be solved so efficiently. 

First we provide an efficient algorithm to find a $2D$ maximum matching 
which satisfies a causality constraint (a maximum weight matching which 
has no intersecting edges). Additionally we will provide an approach 
involving the Hungarian algorithm to discover a maximum weighted $2D$-matching, 
which does not obey the causality constraint (edges involved in the maximum matching 
may intersect).

Given the two time lists $L_1 = \{ t_1, t_2, \ldots, t_n \}$ and 
$L_2 = \{ s_1, s_2, \ldots, s_m \}$ and a general scoring function $f(\cdot)$ over 
the specified time interval $[ \tau_{min, \tau_{max}}]$ we would like to 
find a maximum weighted $2d$ matching between these two time lists, such that 
the matching has no intersecting edges. No intersecting edges intuitively 
guaranties the causality constraint. 
To solve this problem we will employ the dynamic programming approach. 
Let $M_{i,j}$ be a maximum matching with the respective weight 
$w(M_{i,j})$, obeying the causality constraint, involving 
up to and including the $t_i$'th item of the list $L_1$ and up to and including 
the $s_j$'th item in the list $L_2$. Thus, the matching $M_{n,m}$ will hold the 
maximum weighted matching for the entire lists $L_1$ and $L_2$. When we 
compute the matching, we attempt to improve it 
from step to step by adding only the edges(matches) which do not intersect any of the 
edges already present in the matching. The following description of the algorithm 
will show why it is the case.

We will illustrate now that if we have correct solutions to 
subproblems $M_{i-1,j}$, $M_{i,j-1}$ and $M_{i-1,j-1}$, then we 
can construct a maximum matching $M_{i,j}$, which obeys the causality constraint 
by considering the following two simple cases: 

\begin{enumerate}

\item{Either the elements $t_i$ and $s_j$ are both matched to each other in the matching $M_{i,j}$, 
in which case $M_{i,j} = M_{i-1,j-1} \cup (t_i, s_j)$. Obviously the 
edge $(t_i, s_j)$ does not intersect any of the previous edges of $M_{i-1, j-1}$ so we maintain 
the causality constraint;}

\item{Or, the elements $t_i$ and $s_j$ are not matched to each other in the matching $M_{i,j}$. 
Then, one of the $t_i$ or $s_j$ is not matched (see Lemma~\ref{lemma_match}), 
which means that $M_{i,j} = max \{ M_{i,j-1}, M_{i-1,j} \}$. No edges are added to the matching 
in this case.}

\end{enumerate}

We initialize our algorithm by computing in linear time the base set of matches 
$\{ M_{1,1}, M_{1,2}, \ldots , M_{1, n} \}$ 
(the bottom row) and $\{ M_{1,1}, M_{2,1}, \ldots , M_{m, 1} \}$ (the left most column) 
of the two-dimensional array of subproblems (of size  $n \cdot m$) that is being built up. 
The matchings $\{ M_{1,1}, M_{1,2}, \ldots , M_{1, n} \}$ are constructed by taking the first 
element $s_1$ from the list $L_2$ and computing all of the weights of the edges $w(t_i, s_1)$, 
s.t. $w(M_{1,1}) = \{ f(s_1 - t_1) \}$ (contains edge $(t_1, s_1)$, if its not $0$), 
$w(M_{1,2}) = max \{ f(s_1 - t_1), f(s_1 - t_2) \}$ (contains the heavier of two edges $(t_1, s_1)$, 
$(t_2, s_1)$ ) up to $M_{1,n} = max \{ f(s_1 - t_1), f(s_1 - t_2), \ldots, f(s_1 - t_n) \}$ 
(contains the edge of maximum weight considered over all $t_i$'s). 
We similarly compute the set of matchings $\{ M_{1,1}, M_{2,1}, \ldots , M_{m, 1} \}$. 
Next we are ready to fill in the rest of the two-dimensional array of subproblems 
starting with $M_{2,2}$, since $M_{1,1}$, $M_{1,2}$ and $M_{2,1}$ are all available. 
The pseudo code of the algorithm is given in Figure~\ref{fig:algor10}. 

\begin {figure}
\centering
\hspace*{-0.1in}

\begin{algorithmic}[1]
\STATE {\bf Algorithm Match-Causality}
\STATE {Compute $\{ M_{1,1}, M_{1,2}, \ldots , M_{1, n} \}$ 
and $\{ M_{1,1}, M_{2,1}, \ldots , M_{m, 1} \}$}
\FOR   {$i = 2;$ $i \leq n;$ $i++$}
\FOR   {$j = 2;$ $j \leq m;$ $j++$}
\STATE {$M_{i,j} = max \{w(M_{i-1,j-1} \cup (t_i, s_j)), w(M_{i-1,j}), w(M_{i,j-1})  \}$}
\STATE {Store a direction for backtracking}
\ENDFOR
\ENDFOR

\STATE {Start at $M_{m, n}$ and backtrack to retrieve the edges of the matching}

\end{algorithmic}
\caption{Algorithm to discover a maximum weighted matching which obeys the causality constraint. 
In the algorithm above, we initialize $i = 0; j = 0$ ($i,j$ are 
time positions in lists  $L_1 = \{ t_1, t_2, \ldots, t_n \}$ and 
$L_2 = \{ s_1, s_2, \ldots, s_m \}$.}
\label{fig:algor10}
\end{figure}

\begin{lemma}
\label{lemma_match1}
The matching constructed by algorithm Match-Causality, obeys the causality constraint 
(contains no intersecting edges).
\end{lemma}

\begin{proof}
By construction of our algorithm, during the computation of every 
$M_{i,j}$ a new edge is added to the matching only if 
the $(M_{i-1,j-1} \cup (t_i, s_j))$ is picked as maximum. But 
since $t_i$ and $s_j$ are the very last two elements for the matching $M_{i,j}$, they can't 
intersect any of the edges. Thus, since at each step our algorithm 
consistently adds edges which do not intersect any of the previously 
added edges, the final matching will contain no intersecting edges.
\end{proof}

\begin{lemma}
\label{lemma_match}
If the items $t_i$ and $s_j$ are not matched to each other in the matching $M_{i,j}$, 
then one of the $t_i$, $s_j$ is not matched at all.
\end{lemma}

\begin{proof}
Let us assume for the sake of contradiction that both $t_i$ and $s_j$ are matched with 
some nodes. This automatically implies that $t_i$ must be matched with some $s_{j'}$, 
which appears before the $s_j$ in the list $L_2$; and $s_j$ is matched with some $t_{i'}$, 
which occurs before the $t_i$ in the list $L_1$. But this means that the edges $(t_i, s_{j'})$ 
and $(t_{i'}, s_{j})$ intersect, a contradiction.
\end{proof}

\begin{theorem}
Algorithm Match-Causality correctly finds a maximum weighted matching.
\end{theorem}

\begin{proof}
Proof by induction. For the base case lets consider the case where $\|L_1\| = 1$ and 
$\|L_2\| = 1$, in this case the algorithm will trivially match $t_1$ (the only element of $L_1$) 
with $s_1$ (the only element of $L_2$) as long as the $f(s_1 - t_1) > 0$, otherwise the 
matching would be empty.

For the inductive step we assume that if our algorithm finds all of the maximum weighted matchings 
,which obey the causality constraint, 
correctly up to and including $M_{i,j-1}$, then the 
algorithm correctly finds the maximum matching which obeys the causality constraint 
for $M_{i, j}$ (the very next position it considers after $M_{i, j-1}$). 
By our assumption we know that our algorithm correctly found the matchings $M_{i,j-1}$, $M_{i-1,j-1}$ 
and $M_{i-1,j}$, which all obey the causality constraint, 
since all of them occurred before the computation of $M_{i,j}$. 
If so, then our algorithm by construction 
will pick the maximum weight matching 
from the set of $3$ possible matchings 
$\{(M_{i-1,j-1} \cup (t_i, s_j)), M_{i-1,j}, M_{i,j-1}  \}$, 
which guaranties the $M_{i, j}$ to be maximum weight and obey the causality 
constraint.
\end{proof}

\begin{theorem}
Algorithm Match-Causality runs in $O(n \cdot m)$ time.
\end{theorem}


The general propagation delay function $f(\cdot)$ can have any shape, and 
one can wonder if it is possible to find an algorithm 
which will perform faster then $O(n \cdot m)$ for some special case of 
the general propagation delay function. Let us consider one of the 
most intuitive scenarios where the propagation delay function 
is monotonically decreasing. 
We prove that there does not exist an algorithm which can construct the 
maximum weight matching in less then $O(n \cdot m)$ time, which obeys the 
causality constraint.

\begin{theorem}
Algorithm which finds exactly the maximum weight matching 
for a propagation delay function which is 
strictly monotonically decreasing (not a ``step'' function) 
and obeys the causality constraint, 
requires at least $O(n \cdot m)$ time.
\end{theorem}

\begin{proof}
Consider the two time lists 
$L_1 = \{ t_1, t_2, \ldots, t_n \}$,  
$L_2 = \{ s_1, s_2, \ldots, s_m \}$, where 
every time $s_j > t_n$, and a strictly 
monotonically decreasing function $f(\cdot)$, s.t. 
$f(s_m - t_1) > 0$. 
The first observation to make is that 
$t_n$ must be a part of the matching. 
If $t_{n'}$ is the last matched item  
and $t_{n}$ is not matched, where $n' < n$, 
then the matching can be improved by 
replacing $t_{n'}$ with $t_n$, since $f(\cdot)$ 
is a strictly monotonically decreasing function 
and $n' < n$. 

If the matching obeys the causality constraint, 
then the maximum weight matching can be 
$\{ f(s_1 - t_n) \}$ or $\{ f(s_2 - t_n) + f(s_1 - t_{n-1}) \}$ or  
$\ldots$ or $\{ f(s_1 - t_1) + f(s_2 - t_2) + \ldots + f(s_m - t_n) \}$, 
order of $O(n \cdot m)$ combinations. 
And since the function is any strictly monotonically decreasing function, 
one can't guaranty the optimality of the discovered matching without 
having to consider all of the mentioned $O(n \cdot m)$ permutations. 
Thus an algorithm which finds exactly the maximum weight matching 
for a propagation delay function which is 
strictly monotonically decreasing (not a ``step'' function) 
and obeys the causality constraint, 
requires at least $O(n \cdot m)$ time.
\end{proof}

Additionally we present a method to discover a maximum weight matching 
for a general propagation delay function, which doesn't have to obey the 
causality constraint (we allow the intersection of edges in the matching).
The general idea is to use a Hungarian algorithm to 
find a maximum weighted $2d$-matching for a pair of time lists.
  
First, given two time lists $L_1 = \{ t_1, t_2, \ldots, t_n \}$ and 
$L_2 = \{ s_1, s_2, \ldots, s_m \}$ and a general scoring function $f(\cdot)$ over 
the specified time interval $[ \tau_{min, \tau_{max}}]$, we construct the 
bipartite graph, where on the left we have 
the set of $n$ nodes, where each node represents a respective time from 
$\{ t_1, t_2, \ldots, t_n \}$ and on the right we have a set of $m$ nodes 
representing each of $\{ s_1, s_2, \ldots, s_m \}$ times respectively. Each 
pair of nodes $t_i$ and $s_j$ is connected by an edge, where the weight 
on the edge equals to $f(s_j - t_i)$ ($0$ if outside 
the $[ \tau_{min, \tau_{max}}]$ bounds). 

Once we have constructed the bipartite graph we are ready to run the Hungarian 
algorithm. The produced matching $M$ is of maximum weight, but does not take into 
account the causality constraint (some of the edges of $M$ may intersect). This 
algorithm runs in cubic time.

We use ENRON data to test general propagation delay functions against the ``step'' 
function. The results of our experiments are 
presented in Section~\ref{sec:Experimental Results and Conclusions}. 
It turns out that in most of the cases there is not much added 
value from the more general propagation delay function 
in practice. Thus, the more efficient function seems adequate.

\section{Statistically Significant Triples}
\label{sec:Statistically Significant Triples}

We determine the minimum frequency $\kappa$ that makes a triple 
statistically significant, 
using  a statistical model that mimics certain features of the data:
we model the 
inter-arrival time distribution and receiver id probability 
conditioned on sender id, 
to  generate synthetic data and find all 
randomly occurring triples to determine the threshold frequency $\kappa$.

\subsection{A Model for the Data}

We estimate directly from the data the message inter-arrival time 
distribution $f(\tau)$, the conditional probability distribution 
$P(r|s)$, and the marginal distribution \math{P(s)}
using simple histograms (one for $f(\tau)$, $S$ for $P(r|s)$ and
\math{S} for 
$P(s)$, i.e. one conditional and marginal distribution 
histogram for each sender, where \math{S} is the number of senders). 
One may also model additional features 
(e.g. $P(s|r)$), to obtain more accurate models. One 
should however bear in mind that the more accurate the model, the 
closer the random data is to the actual data, hence the less useful 
the statistical analysis will be - it will simply reproduce the data.

\subsection{Synthetic Data}
 Suppose one wishes to generate $N$ messages using $f(\tau)$, 
$P(r|s)$ and $P(s)$. First we generate $N$ inter-arrival 
times independently, which specifies the times of the communications. 
We now must assign 
sender-receiver pairs to each communication. The senders are selected 
independently from $P(s)$. 
We then generate each receiver independently, but conditioned on the
sender of that communication, according to $P(r|s)$.

\subsection{Determining the Significance Threshold}
To determine the significance threshold $\kappa$, 
we generate $M$ (as large as possible) synthetic data sets and 
determine the 
triples together with their frequencies of occurrence in each 
synthetic data set. The threshold $\kappa$ may be selected as 
the average plus two standard deviations, or (more conservatively) 
as the maximum frequency of occurrence of a triple.

\section{Constructing Larger Graphs using Heuristics}
\label{section:heuristics}

Now we discuss a heuristic method for building larger communication 
structures, using only statistically significant triples.  
We will start by introducing the notion of an overlap factor. 
We will then discuss how the overlap factor is used to build a larger 
communication graph by finding clusters, and construct the larger 
communication structures from these clusters.

\subsection{Overlap between Triples}

For two  statistically significant triples $(A,B,C)$ and $(D,E,F)$ 
(chain or sibling) with 
maximum matchings at the times $M_1 = \{ (t_1,s_1), \ldots, (t_k,s_k)\}$ 
and $M_2 = \{ (t'_1,s'_1), \ldots, (t'_p,s'_p)\}$,
we use an overlap weighting function $W(M_1,M_2)$ to capture the 
degree of coincidence between the matchings $M_1$ and $M_2$. 
The simplest such overlap weighting function is the extent to which 
the two time intervals of communication overlap. Specifically, 
$W(M_1, M_2)$ is the percent overlap between the two intervals 
$[t_1,s_k]$ and $[t'_1, s'_p]$: 
$$
W(M_1,M_2) = \max \left\{ \frac{\min(s_k,s'_p) - 
\max(t_1,t'_1)}{\max(s_k,s'_p) - \min(t_1,t'_1)}, 0 \right \}
$$
A large \emph{overlap factor} suggests that both 
triples are part of the same hidden group. More 
sophisticated overlap factors could take into account intermittent 
communication but for our present purpose, we will
use this simplest version.

\subsection{The Weighted Overlap Graph and Clustering}

We construct a weighted graph by taking all significant triples to 
be the vertices in the graph. Let $M_i$ be the maximum matching 
corresponding to vertex (triple) $v_i$. 
We define the weight of 
the edge $e_{ij}$ to be $\omega(e_{ij}) = W(M_i,M_j)$, producing 
an undirected complete graph (some weights may be 0). 
By thresholding the weights, one could obtain a sparse graph. 
Dense subgraphs 
correspond to triples that were all 
active at about the same time, and are a candidate hidden
group. 
We want to cluster the graph into dense possibly overlapping 
subgraphs. Given the triples in a cluster we can build a directed 
graph, consistent with all the triples,
 to represent its communication structure. Cluster containing 
multiple connected components implies the existence 
of some hidden structure connecting them. Below is an outline of the 
entire algorithm:
\begin{center}
\begin{algorithmic}[1]
\STATE{Obtain the significant triples.}
\STATE{Construct a weighted overlap 
graph (weights are overlap factors between pairs of triples).}
\STATE{Perform clustering on the weighted graph.}
\STATE{Use each cluster to determine a candidate hidden group structure.}
\end{algorithmic}
\end{center}
For the clustering, since clusters may overlap, 
we use the algorithms presented in \cite{cluster2005a, cluster2005b}.
\section{Algorithm for Querying Tree Hidden Groups}

We describe efficient algorithms for computing (exactly)
the frequency
of a hidden group whose communication structure is an arbitrary 
pre-specified
tree. 
We assume that messages initiate from the root. 
The parameters \math{\tau_{min},\tau_{max},\delta}
are also specified. Such an algorithm
can be used in conjunction with the previous heuristic algorithms to 
verify that a discovered tree-like structure actually occurs frequently in 
the data.

Let $L$ be an adjacency list for the  tree \math{T}, $D$ 
a dataset in which we will 
query this tree. 
The first entry in the list $L$ is the {\it root communicator} followed 
by the list of 
all its {\emph children} (receivers) the {\it root} sends to. 
The next entries in $L$ contain the 
lists of children for each of the receivers of $L_{root}$ 
 until we reach the leaves, which have no children.
 \begin{figure}[p]
\begin{center}

\begin{algorithmic}[1]
\STATE{\bf Algorithm Tree-Mine($T$,$D$)}
\STATE{$D_{rem} \gets D$}
\WHILE{$M=TRUE$}
\STATE{($M,t') = FindNext(T,D_{rem}$)}
\IF{$M$}
\STATE{Store Match}
\STATE{Increment List Pointers; get $D_{rem}$}
\ENDIF
\ENDWHILE
\end{algorithmic}

\begin{algorithmic}
 \STATE{  }
 \STATE{  }
  \STATE{  }
 \end{algorithmic}

\begin{algorithmic}[1]
\STATE{\bf Algorithm $FindNext(T,D_{rem}$)}
\STATE{Initialize all $truth_j \gets 0$}
\STATE{return $Findnext_{rec}(NULL,root)$}
\end{algorithmic}

\begin{algorithmic}
 \STATE{  }
 \STATE{  }
  \STATE{  }
 \end{algorithmic}

\begin{algorithmic}[1]
\STATE{\bf Algorithm $FindNext_{rec}(t,\ast node$ $i)$}
\STATE{$(\star)$Run Algorithm-Sibling from current
time list pointers to get $m = (t_1,\ldots,t_n)$}
\IF{$m \sim t$}
\FOR{$j = n$ to 1}
\IF{($truth_j=1$ $\&$ $prev_j<t_j$) or $truth_j=0$}
\STATE{$(truth_j,prev_j)=FindNext_{rec}(t_j,\ast node$ $j)$}
\IF{$truth_j=0$}
\STATE{Increase $t_j$ pointer, GOTO$(\star)$}
\ENDIF
\ELSE
\STATE{return ($truth_j$,t)}
\ENDIF
\ENDFOR
\ENDIF
\IF{$m < t$}
\STATE{Increase $t_j$ pointer, GOTO$(\star)$}
\ENDIF
\IF{$m > t$}
\STATE{return ($0$, t)}
\ENDIF
\end{algorithmic}
\caption{Algorithms used for Querying a Tree $T$ in the data $D$. 
In the algorithms above, $D_{rem}$ represents $D$ in an way that 
allows the Tree-Mine Algorithm to efficiently access the necessary data.}
\label{fig:tree_mining}
\end{center}
\end{figure}

After we have read in $D$, 
we process $L$ and use it to construct the tree, 
in which every node will contain:
\emph{node id}, \emph{time list} when its parent  
sent messages to it, and a \emph{list of children}. 
We construct such tree by processing $L$ 
and checking each communicator that has children if it is present in 
$D$ as a $Sender$, 
and if its children are present in $D$ in the list of its $Receivers$.
During the construction, if a node that has children is 
not present in $D$ as a $Sender$, or some child is not on the list 
of $Receivers$ of its parent, then we know that the given tree does not 
exist in the current data set $D$ and we can stop our search.

For a tree to exist there should be at least one matching involving 
all of the nodes (lists). We start with $root$ and consider the 
time lists of its children. We use Algorithm-Sibling to find the first matching 
$m_1=(t_1, t_2, \ldots, t_n)$, where $t_i$ 
is an element of the $i$'s child time list and 
$n$ is the number of time lists. After the first matching $m_1$ 
we proceed by considering the children 
in the matching $m_1$ from the rightmost child to the 
left by taking the value $t_i$, 
which represents this node in the matching
and passing it down to the child node. 
Next we try to find a matching $m_2 = (s_1, s_2, \ldots, s_k)$ 
for the \math{k} child time lists. 
There are three cases to consider:

\begin{enumerate}

\item{ Every element $s_j$ of the matching $m_2$ also satisfies 
the chain-constraint with the element $t_i$: 
$\tau_{min} \leq s_j - t_i \leq \tau_{max}$, 
$\forall s_j \in m_2$, $j  = k, \ldots, 1$. 
In this case we say $m_2 \sim t_i$ ($m_2$ matches $t_i$) 
and proceed by considering all children. 
Otherwise consider the rightmost $s_j \in m_2$. The two cases below refer 
to \math{s_j}.}

\item{ If $s_j < t_i + \tau_{min}$, in which case we say $m_2 < t_i$,
we  advance to the next 
element in child $j$'s time list and continue 
as with Algorithm-Sibling 
to find the next matching $(s'_1, s'_2, \ldots, s'_k)$. 
This process is repeated as long as $m_2 < t_i$. Eventually we will find 
an $m_2$ with 
$m_2 \sim t_i$ or we will reach the end of some list (in which case 
there is no further 
matching) or we come to a matching $m_2 > t_i$ (see case below).}

\item{ If $s_j > t_i + \tau_{max}$, in which case we say $m_2 > t_i$,
we advance $t_i$ to the next element in $i$'s time list on the previous 
level and proceed as with Algorithm-Sibling to find the next matching in the
previous level. 
After this new matching $(t'_1, t'_2, \ldots, t'_n)$ is found, the 
chain 
constraints have to be checked for these time 
lists $(t'_1, t'_2, \ldots, t'_n)$ 
with their 
 previous level and the algorithm proceeds recursively from then on.}

\end{enumerate}

 \begin{figure}[htpb]
  \begin{center}
   \resizebox{5cm}{!} {\input{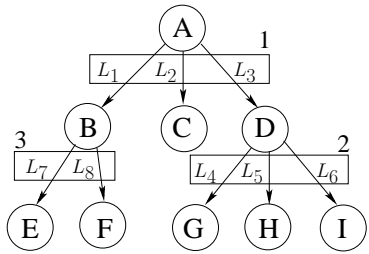}}
  \end{center}
  \caption{\small Example of a communication tree structure}
  \label{fig:mining_example}
\end{figure}

The entire algorithm for finding a complete matching 
can be formulated
into two steps: find the first matching; recursively process the remaining 
parts of the time lists. What we have described is the 
first step which is accomplished 
by calling the recursive algorithm $FindNext_{rec}(NULL,root)$ 
that is summarized in the Figure \ref{fig:tree_mining}. 
If this returns $TRUE$,
 the algorithm has found the 
first occurrence of the tree, 
which can be read off from
 the current time list pointers.
After this instance is found, we store it and proceed by considering  
the remaining part of the time lists starting from the $root$.

To illustrate how the Algorithm Tree-mine works, consider 
the example tree $T$ in Figure \ref{fig:mining_example}. Let 
node $A$ to be a $root$ and let $L_1, \ldots, L_8$ be the time lists. 
Refer to $(L_1, L_2, L_3)$ as the $phase_1$ lists,
 $(L_4, L_5, L_6)$ 
as the $phase_2$ lists and  $(L_7, L_8)$ as the $phase_3$ lists. 
Let $m_1 = (t_1, t_2, t_3)$, $m_2 = (s_1, s_2, s_3)$ and 
$m_3 = (r_1, r_2)$ 
be the first matchings of the $phase_1, phase_2$ and $phase_3$ lists
 respectively. 
If $m_2 \sim t_3$ and $m_3 \sim t_1$, we have found the first 
matching and we now 
recursively process the remaining time lists. If 
$m_2 < t_3$ (eg. $s_2 < t_3 + \tau_{min}$), 
then we move to the next matching in the $phase_2$ lists. 
If $m_2 > t_3$ then we move to the next matching in 
$phase_1$ lists and reconsider the 
$phase_2$ matching and the $phase_3$ matching if necessary. If $m_2 \sim t_3$ 
we then similarly check $m_3$ with $t_1$. Since node 
$C$ is a leaf, it need not be further processed.

\begin{theorem}
Algorithm Tree-Mine correctly finds the maximum number of occurrences for a specified 
tree $T$.
\end{theorem}

\begin{proof}
Proof by contradiction. 
Given a set of time lists $L=(L_1, L_2, \ldots, L_n)$ that 
specify a tree, our algorithm produces a matching 
$M=(m_1, m_2, \ldots, m_k)$, where each matching $m_i$ is a sequence of 
$n$ times from each of the $n$ time 
lists $m_i=(t_1^{i}, t_2^{i}, \ldots, t_n^{i})$. 
Let $M^{*}=(m^{*}_1, m^{*}_2, \ldots, m^{*}_{k^{*}})$ be a maximum 
matching of size $k^*$.
The next lemma follows directly from the construction of the 
Algorithms.

\begin{lemma}
\label{lemma4}
If there is a valid matching our algorithm will find one.
\end{lemma}

\begin{lemma}
\label{lemma5}
Algorithm Tree-Mine finds an earliest valid matching(occurrence). 
Let the first valid matching found by our algorithm 
be $m_1 = (t_1, t_2, \ldots, t_n)$, then for any other valid 
matching $m' = (s_1, s_2, \ldots, s_n)$ $t_i \leq s_i$ 
$\forall$ $i = 1,\cdots,n$.
\end{lemma}

\begin{proof}
Proof by contradiction. 
Assume that in $m_1$ and $m'$ there exists a corresponding pair of times $s < t$ 
and let $s_i$, $t_i$ be the first such pair. 
Since $m_1$ and $m'$ are valid matchings, then $s_i$ and $t_i$ obey the chain 
constraints: 
$ \tau_{min} \leq (t_i - t_p) \leq \tau_{max}$ (where $t_p$ is a time 
passed down by the parent node), 
$\tau_{min} \leq ( (t_{children} - t_{i}) \leq \tau_{max}$ (where $t_{children}$ are 
times of children of $t_i$), similarly
 $\tau_{min} \leq (s_i - s_p) \leq \tau_{max}$, 
$\tau_{min} \leq (s_{children} - s_{i}) \leq \tau_{max}$; and obey the 
sibling constraint: 
$ \tau_{min} \leq (t_{i+1} - t_i) \leq \tau_{max}$, 
$\tau_{min} \leq (t_i - t_{i-1}) \leq \tau_{max}$ 
(where $t_{i-1}$ and $t_{i+1}$ are matched times of neighboring siblings of $t_i$)
and similarly $\tau_{min} \leq (s_{i+1} - s_i) \leq \tau_{max}$, 
$\tau_{min} \leq (s_{i} - s_{i-1}) \leq \tau_{max}$.

Since $s_i < t_i$, then $\tau_{min} < (t_{i+1} - s_i)$ and $\tau_{max} > (s_i - t_{i-1})$. 
Also because $s_{i-1} \geq t_{i-1}$, we get that $ \tau_{min} \leq (s_i - t_{i-1})$ and 
since $(s_{i+1} - s_i) \leq \tau_{max}$, then $(min(t_{i+1},s_{i+1}) - s_i) \leq \tau_{max}$ 
as well. By similar reasoning since $s_i < t_i$, then 
$\tau_{min} < (t_{children} - s_i)$ and $\tau_{max} > (s_i - t_{p})$; also 
since $s_{p} \geq t_{p}$, we get that $ \tau_{min} \leq (s_i - t_{p})$ and 
since $(s_{children} - s_i) \leq \tau_{max}$, then $(min(t_{children},s_{children}) - s_i) \leq \tau_{max}$ 
as well. But if $s_i$ satisfies all of the chain and sibling constraints, 
then $m_1$ would not be the first valid matching as has already been proven for 
algorithms chain and sibling triples, 
because the first matching $m_f$ would contain 
$m_f = (t_1, t_2, \ldots , t_{i-1}, s_i, min(t_{i+1}, s_{i+1}), min(t_{i+2}, s_{i+2}), \ldots , min(t_{n}, s_{n}))$. 
Thus, algorithm Tree-Mine finds the earliest possible matching(occurrence). 
\end{proof}

Now let us for the purpose of contradiction assume that Tree-Mine does not find a maximum 
number of occurrences of a specified tree $T$, s.t. $k < k^{*}$.


The situation where $k < k^{*}$ can only appear if $M^{*}$ discovers an occurrence of $T$ before 
the Tree-Mine does, s.t. some occurrence $m_i^{*}$ which is earlier then its respective occurrence $m_i$. 
But such a situation can not happen, since, given the set of time lists (or the remainder of them, if we 
already processed some of them) which define the tree $T$, Tree-Mine guaranties to find the earliest 
valid match by Lemma~\ref{lemma5}. Thus, we obtain a contradiction. 
This proves that Tree-Mine correctly finds the maximum number of occurrences of a specified 
tree $T$.
\end{proof}

\begin {theorem} 
Algorithm Tree-Mine runs in $O(d_{max} \cdot \|D\|)$.
\end{theorem}


\subsection{Mining all Frequent Trees}

Here we propose an algorithm which allows us to discover 
the frequency of general trees and to enumerate all statistically 
significant general trees of a specified size and frequency. 
The parameters \math{\tau_{min},\tau_{max}, \delta} and $\kappa$ must be specified. 
Additionally you can specify the min and the max tree size to bound  
the size of the trees of interest. The parameter $\kappa$ in this algorithm 
represents the minimal frequency threshold, 
and is used to discard the trees which occur 
fewer times then the specified threshold.

As for any tree mining problem, there are two main steps 
for discovering frequent trees. 
First, we need a systematic way of generating 
candidate trees whose frequency is to be computed. 
Second, we need efficient ways of counting 
the number of occurrences of each
candidate in the database $D$ and determining which
candidates pass the threshold. 
To address the second issue we use our \emph{ Algorithm Tree-Mine} 
to determine the frequency of a particular tree. 
To systematically generate new candidate trees we inherit the idea of an 
\emph{Equivalence Class-based Extensions} and the \emph{Rightmost Path Extensions} proposed and described in
\cite{mining2005a, mining2005b}. 

The algorithm proceeds in the following order: 

(\rn{1})
Systematically generate new candidates, by extending only the frequent trees 
until no more candidates can be extended;

(\rn{2})
Use \emph{ Algorithm Tree-Mine} to determine the frequency of our candidates; 

(\rn{3})
If the candidate's frequency is above threshold - store the candidate.

The main advantage of equivalence class extensions is that only known frequent 
elements are used for extensions. But to guaranty that all possible extensions 
are considered, the non-redundant tree generation idea has to be relaxed. In this way 
the canonical class (considers candidates only in canonical form) and equivalence class 
extensions represent a trade-off between the number of 
isomorphic candidates generated and the number of potentially frequent candidates to count.

\begin {theorem} 
Mining all of the tress on the current level requires $O(n^2 \cdot d_{max} \cdot \|D\| \cdot (v + log(d_{max})))$ operations.
\end{theorem}



\section{Comparing Methods}
\label{sec:Measuring Similarity between Sets of Overlapping Clusters}

To compare methods, we need to be able to measure similarity 
between sets of overlapping clusters. We will use the {\it Best Match} 
approach proposed in \cite{similarity2010}, which we briefly describe here.

We formally define the problem as follows: 

\begin{itemize}
\item Let $C_1 = \{S_1, S_2, \ldots, S_n \}$ and $C_2 = \{ S_{1}', S_{2}', \ldots, S_{m}'\}$ 
be the two clusterings of size $n$ and $m$ respectively, 
where $S_{i}$ and $S_{j}'$ 
are the groups that form the clusterings. 
A group does not contain duplicates.

\item Let $D_{(C_{1}, C_{2})}$ be the distance, 
between the clusterings $C_{1}$ and $C_{2}$.
    
\item The task is to find $D_{(C_{1}, C_{2})}$ efficiently, 
while ensuring that $D_{(C_{1}, C_{2})}$ reflects 
the actual distance between the network 
structures that $C_{1}$ and $C_{2}$ represent.    
\end{itemize}

The {\it Best Match} algorithm determines how 
well the clusterings represent each other. 
That is when given $C_1 = \{S_1, S_2, \ldots, S_n \}$ 
and $C_2 = \{ S_{1}', S_{2}', \ldots, S_{m}'\}$ 
it will determine how well $C_2$ represents 
$C_1$ and vice-versa. 

We begin by considering every group $S \in C_1$ and 
finding a group $S' \in C_2$ with 
the min distance $d_{(S, S')}$ between them. 
The best match algorithm can run with 
any set difference measure which measures the distance 
between two sets $S$, $S'$. 
We define the distance $d_{(S, S')}$ 
between the two groups $S$ and $S'$ as the 
number of moves(changes) necessary to convert $S$ 
into $S'$: 
\[
d_{(S, S')} = |S| + |S'| - 2|S \cap S'|
\]
Note, that alternatively we can also define $d_{(S, S')}$ as:
\[
d_{(S, S')} = 1 -  \frac{|S \cap S'|}{| S \cup S' | }
\]
As we step through $C_1$, we find 
for each group $S_k \in C_1$ the closest 
group $S_{l}' \in C_2$ with a minimal distance $d_{(S_k, S_l')}$:
\[
d_{(S_k, C_2)} = \min_{l = 1, \ldots, m}(d_{(S_k, S_{l}')})
\]
Next we sum up all such distances. 
For the purposes of normalization 
one can normalize the 
obtained sum by the 
total number of distinct members $T_{C_1}$ in $C_1$ 
to obtain $D_{(C_1, C_2)}$:
\[
D{(C_1, C_2)} = \frac{\sum_{k = 1}^{n}d_{(S_k, C_2)}}{T_{C_1}},
\]
\[
T_{C_1} =  \| \cup_{k = 1}^{n} S_k \|;  
\]
this normalization computes a distance per node. One can 
also normalize by $\|C_1\|$ and $\|C_2\|$.

So far we successfully found the distance 
measure $D_{(C_1, C_2)}$ of how well 
the groups in $C_1$ are represented in $C_2$. 
If this {\it asymmetric} measure of the distance 
is considered adequate, 
one may stop the algorithm here. 
However, since in most of the cases 
we want the measure to be {\it symmetric} 
with respect to both clusterings, 
we also want to know how well $C_1$ represents 
$C_2$. We will thus repeat the same calculation 
for each group in $C_2$ with respect to the 
groups in $C_1$ and normalize the sum of 
distances using one of the normalization methods. 
Finally, the {\it Best Match } symmetric distance between 
a pair of clusterings $C_1$ and $C_2$ defined as:  
\[
D_{BestMatch}(C_1, C_2) = \frac{D_{(C_1, C_2)} + D_{(C_2, C_1)}}{2}
\]

This result can be viewed as a representation of 
the average number of moves per distinct 
member (or set) necessary to represent one clustering by the other.


Intuitively the {\it Best Match} algorithm is a 
relative measure of distance, and reflects how well two 
clusterings represent each other, and how similar/different are the 
social networks formed by these clusterings. 
This approach is not sensitive to having clusterings 
of different size or having overlapping sets. Refer to \cite{similarity2010} 
for more details.

\label{sec:enron-data-main}
\label{sec:enron-data}
\section{Enron Data}

The Enron email corpus consists of emails 
released by the U.S. Department of Justice
during the investigation of Enron.
This data includes about 3.5 million emails sent from and to
Enron employees between 1998 and 2002.
The list of approximately 150 employees mailboxes constitute the Enron dataset. 
Although the dataset contains emails related to thousands of Enron
employees, the complete information is only known for this smaller
set of individuals. 
The corpus contains detailed information about each email,
including sender, recipient(s) (including To, CC, and BCC fields),
time, subject, and message body.
We needed to transform this data into our standard input format 
({\sf sender, receiver, time}).
To accomplish this, for each message we generated multiple entries
({\sf sender, receiver1, time}),
$\ldots$
({\sf sender, receiverN, time}),
for all {\sf N} recipients of the message.

\section{Weblog (Blog) Data}
\label{sec:weblog-data}

This data set was constructed by observing 
Russian {\sf livejournal.com} blogs. 
This site allows any user to  create a blog at no cost.
The user may then submit text messages called {\em posts} onto their
own page.
These posts can  be viewed by anyone who visits
their page.
For the purposes of our analysis we would like to know how information
is being disseminated throughout this blog network.
While data about who accessed and read individual home pages is not available,
there is other information with which we can identify communications.
When a user visits a home page, he or she may decide to leave {\em comments}
on one or more posts on the page.
\begin{wrapfigure}{r}{2.4in}
\vspace*{-0.1in}
\centering
\resizebox{2in}{!}{\includegraphics{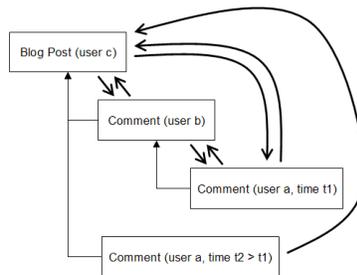}}
\caption{Communications inferred from weblog data.}
\label{fig:comment-links}
\end{wrapfigure}
These comments are visible to everyone, and other users may leave comments
on comments, forming  trees of communications rooted at each post.
We then must process this information
into links of the form ({\sf sender, receiver, time}).
We make the following assumptions for a comment by user $a$ at time $t$
in response to a comment written by user $b$, 
where both comments pertain to a post written by user $c$:
\math{a} has read the original post of \math{c}, hence a 
communication
($c, a, t$) if this was the earliest
comment $a$ made on this particular post.
\math{c} reads the comments that are made
on his site, hence
a communication ($a, c, t$);
\math{a} read the comment to which he is replying,
hence the communication ($b, a, t$);
\math{b} will monitor comments to his comment, hence
the communication ($a, b, t$);
Fig.~\ref{fig:comment-links} shows these assumed communications.
Note that the second post by user $a$ only generates a communication
in one direction, since it is assumed that user $a$ has already read
the post by user $c$.

In addition to making comments, LiveJournal members may select
other members to be in their ``friends'' list.
This may be represented by a graph where there is a directed edge
from user $a$ to user $b$ if $a$ selects $b$ as a friend.
These friendships do not have times associated with them, and so
cannot be converted into communication data.
However, this information can be used to validate our algorithms,
as demonstrated in the following experiment.

The friendship information may be used to verify the groups that have been
discovered by our algorithm.
If the group is indeed a social group, the members should be more
likely to select each other as a friend than a randomly selected group.
The total number of members in the friendship network is 2,551,488, with
53,241,753 friendship
links among them, or about 0.0008 percent of all
possible links are friendship links.
Thus, we would expect about 0.0008 percent of friendship links to be
present in a randomly selected group of LiveJournal members.

\section{Experimental Results}
\label{sec:Experimental Results and Conclusions}

\subsection {Triples in Enron Email Data}
For our experiments we considered the Enron email corpus
(see Section~\ref{sec:enron-data}). 
We took $\tau_{min}$ to be 1 hour and $\tau_{max}$ to be 1 day. 
Fig.~\ref{fig:triples} compares the number of
 triples occurring in the data to the number that occur randomly 
in the synthetically generated  data using the model
derived from the Enron data.
As can be observed, the number of triples 
in the data by far exceeds the random triples. After some frequency 
threshold, no random triples of higher frequency appear - i.e., 
all the triples appearing in the data at this frequency are 
significant. We used $M = 1000$ data sets to determine the random 
triple curve in Fig.~\ref{fig:triples}.

The significance thresholds we discover prove that the probability 
of a triple occurring at random above the thresholds is in practice 
very close to zero. In other words, the observed probability $B$ 
of a random triple occurring above the specified threshold is $0$, 
however the true probability $T$ of a random triple occurring above the 
threshold is not $0$. Thus to put a bound on the 
true probability $T$, we use the Chernoff bound: 
$P( T < \epsilon ) \geq 1 - e^{-2 \cdot n \cdot \epsilon ^2}$,
where $n$ is the number of random sets we generated ($M = 1000$) and 
$\epsilon$ would be an error tolerance. Setting $\epsilon = 0.05$, 
we have that the probability of $P( T < 0.05 ) \geq 0.9933$.

\begin{figure}[ht]
  \centering
  \begin{tabular}{c c}
   \resizebox{7cm}{!} {\label{fig:chains}\includegraphics[scale=1]{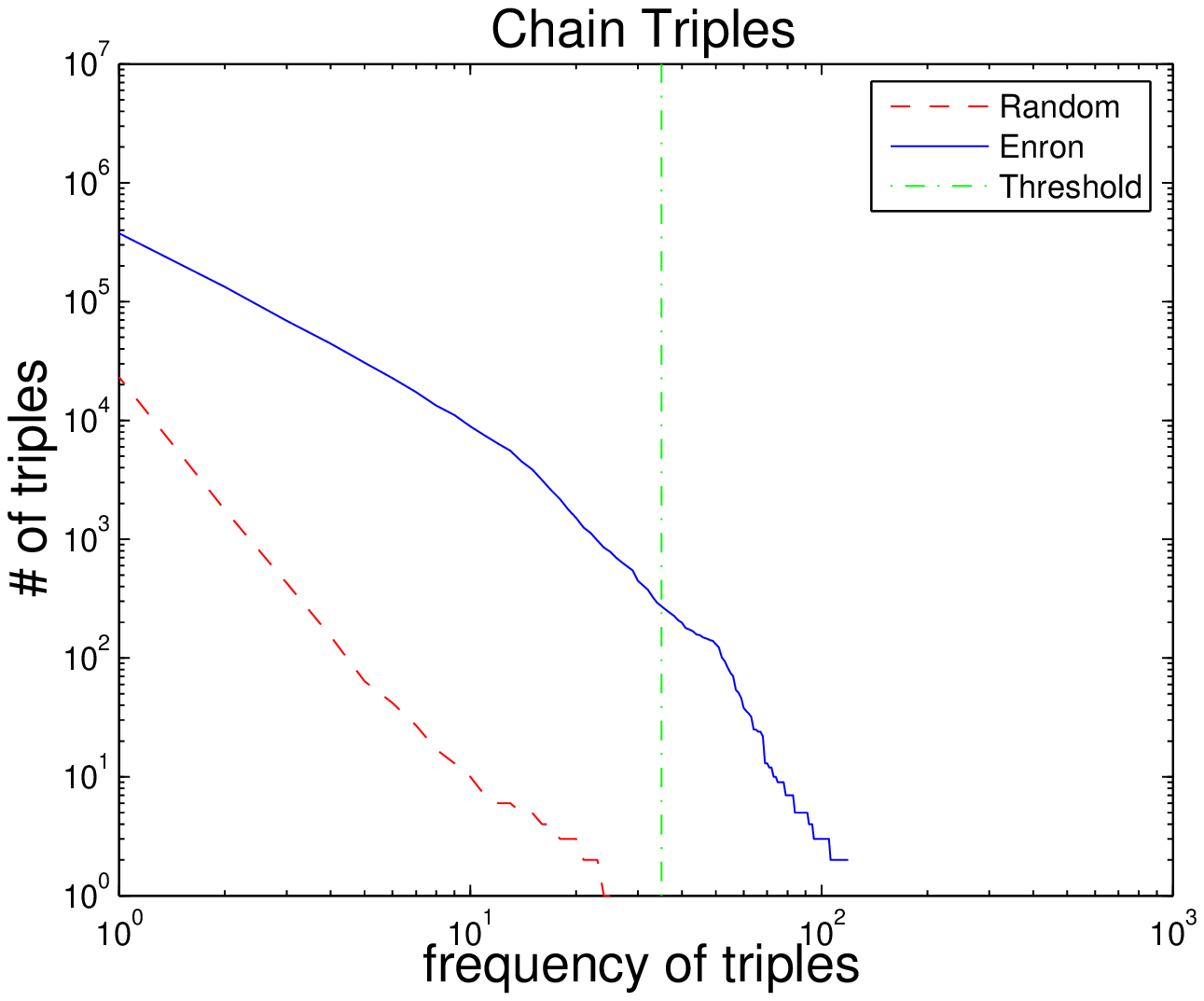}}
   &
   \resizebox{7cm}{!} {\label{fig:siblings}\includegraphics[scale=1]{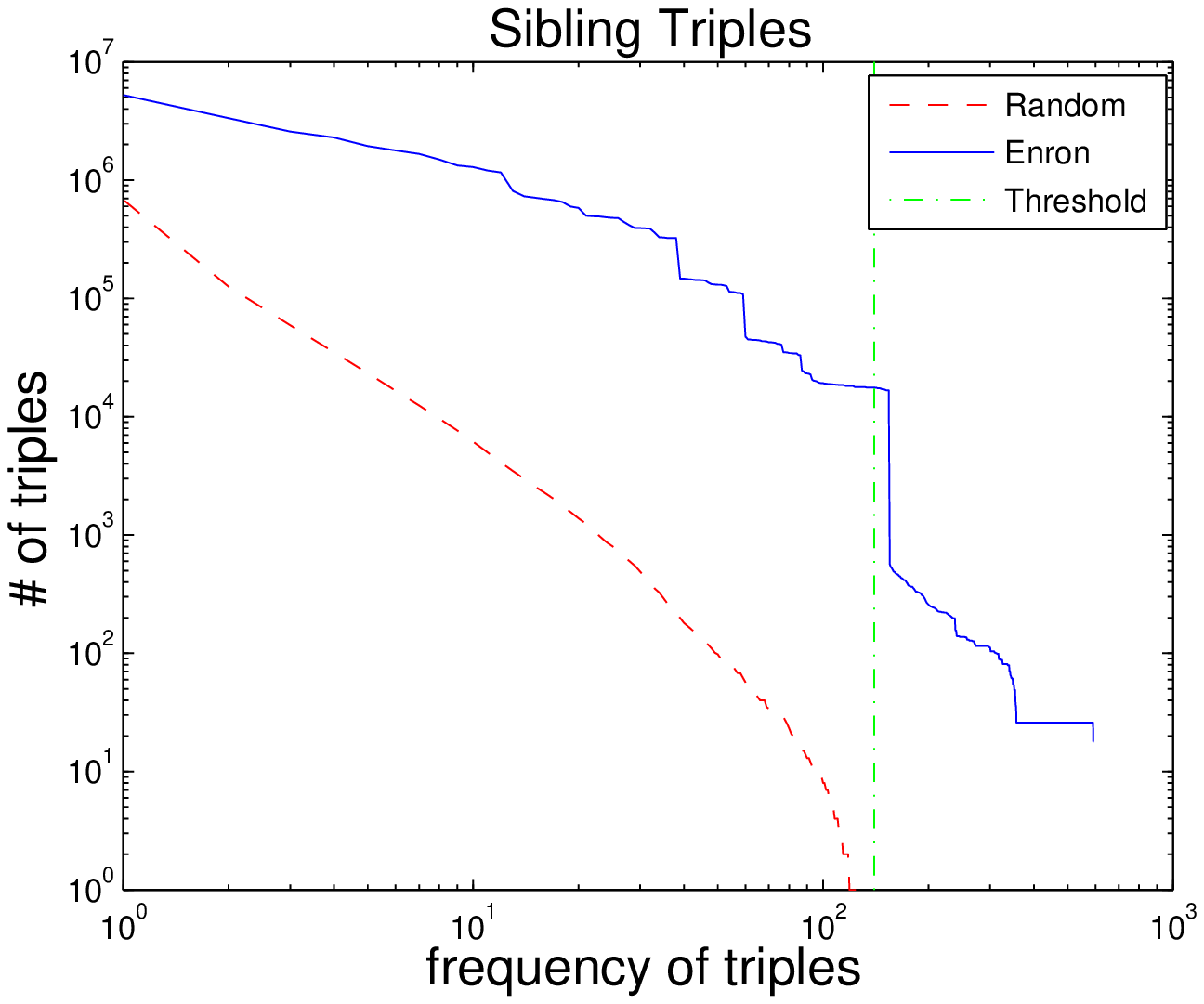}} 
   \\
   (a) & (b)
   \end{tabular}
  \caption{Abundance of triples occurring as a function of frequency of 
occurrence. (a) chain triples; (b) sibling triples }
  \label{fig:triples}
\end{figure}

\begin{figure}[ht]
\centering
\begin{tabular}{c c c}
   \resizebox{3.5cm}{3.2cm}{\label{fig:year_4}\includegraphics{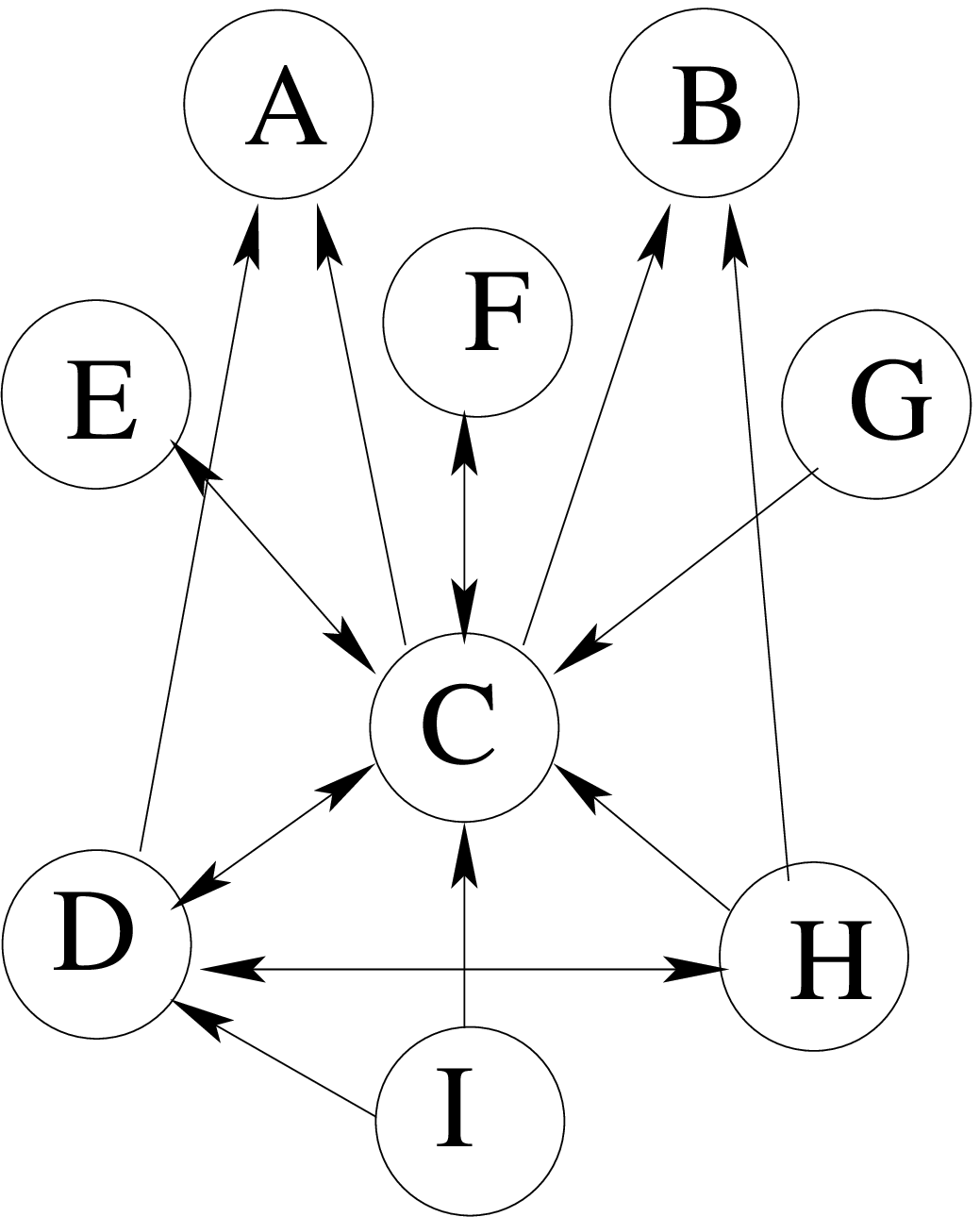}}&
   \resizebox{3.5cm}{3.2cm}{\label{fig:year_5}\includegraphics{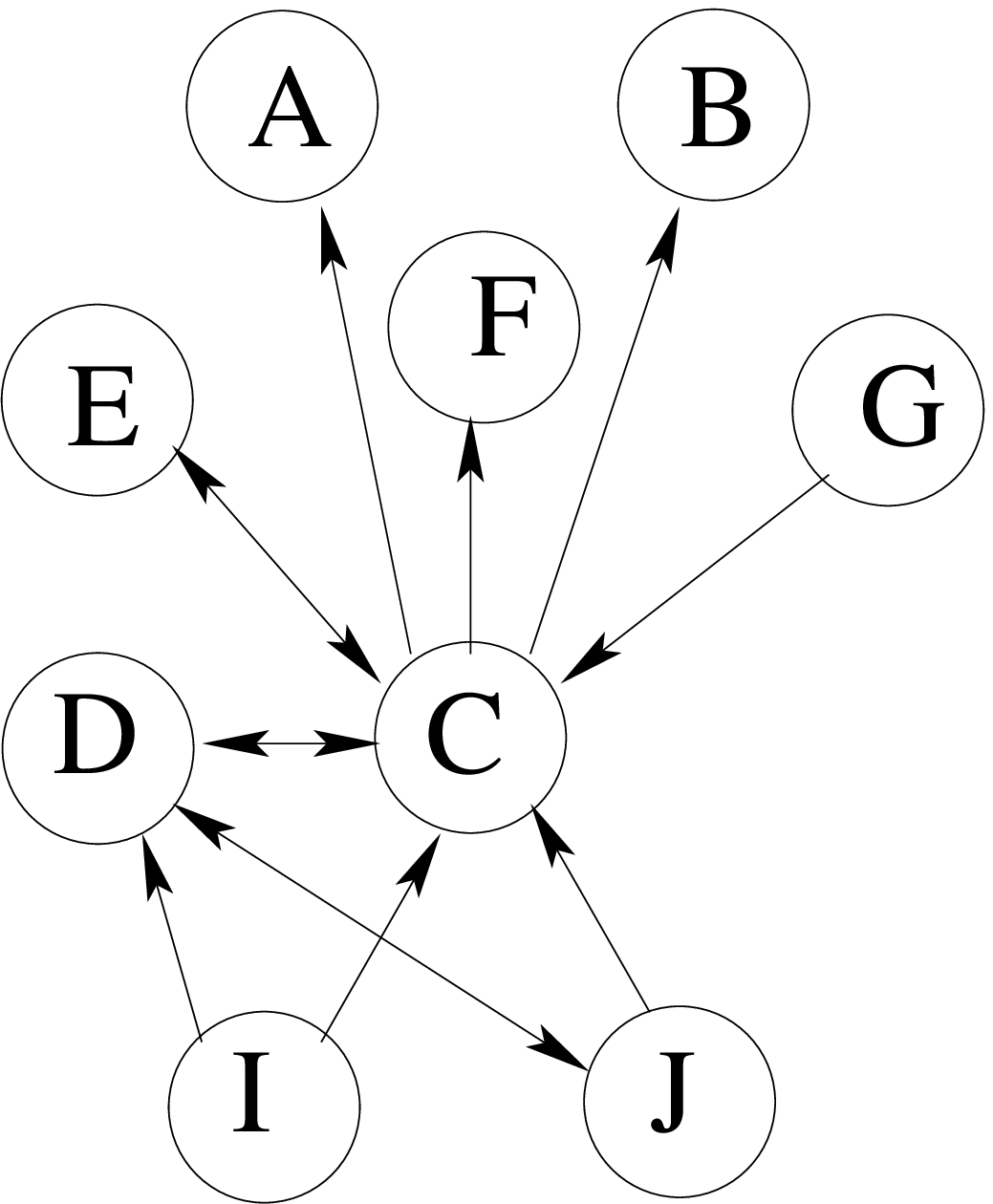}}&
   \resizebox{3.5cm}{3.2cm}{\label{fig:year_6}\includegraphics{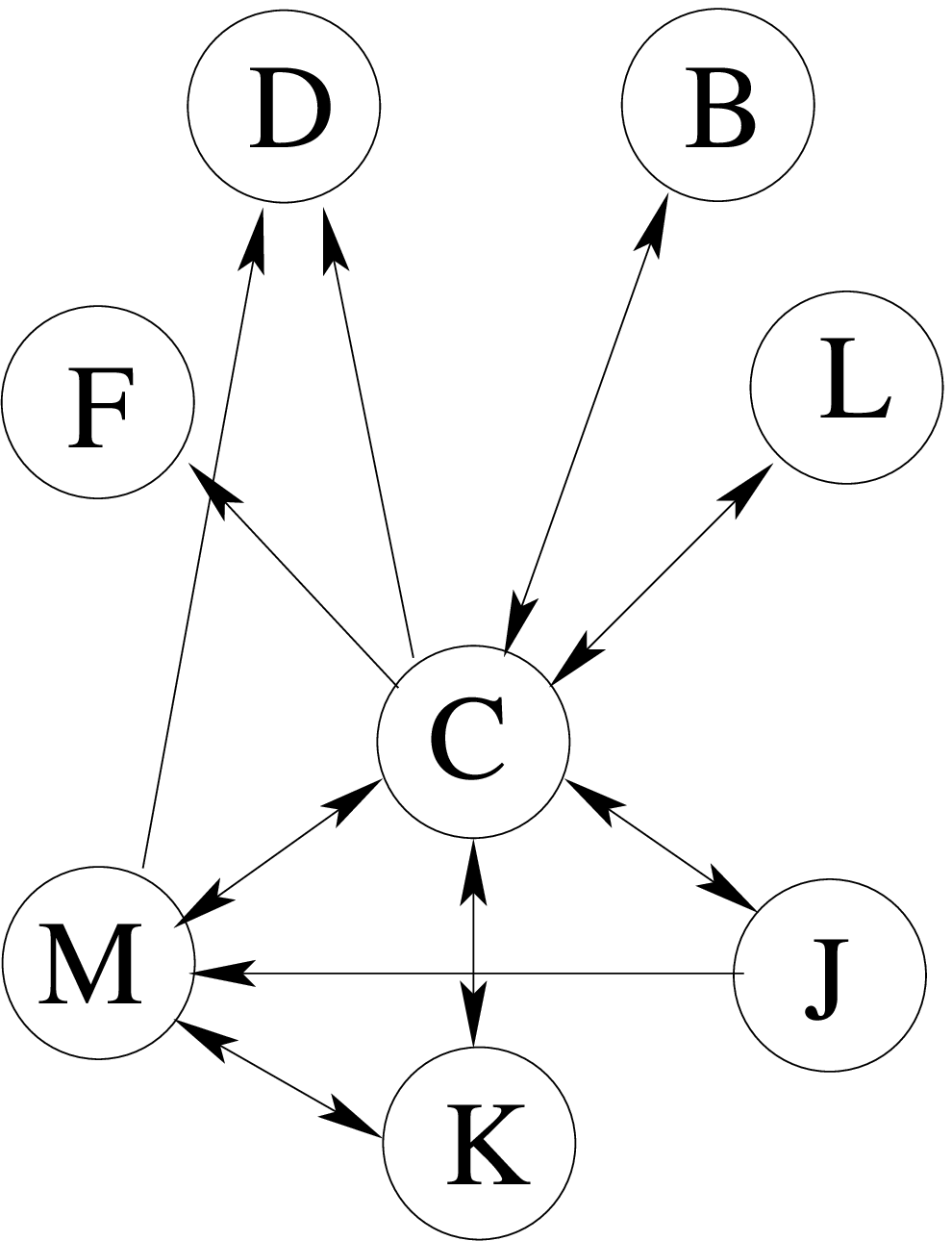}}
   \\
Sept. 2000 - Sept. 2001
& Mar. 2000 - Mar. 2001
& Sept. 2001 - Sept. 2002
\end{tabular}
\caption{Evolution of part of the Enron organizational structure from 2000 - 2002. Note: 
actors $B,C,D,F$ present in all three intervals. Here is who they are: $B$ - T. Brogan, 
$C$ - Peggy Heeg, $D$ - Ajaj Jagsi and $F$ - Thresa Allen.}
\label{fig:enron1}
\end{figure}

\subsection{Experiments on Weblog Data}
Similar experiments were run on the Weblog data to obtain communication groups
(see Section~\ref{sec:weblog-data} for a description of the Weblog data).
As a validation we used a graph of friendship links, which was 
\begin{figure}[ht]
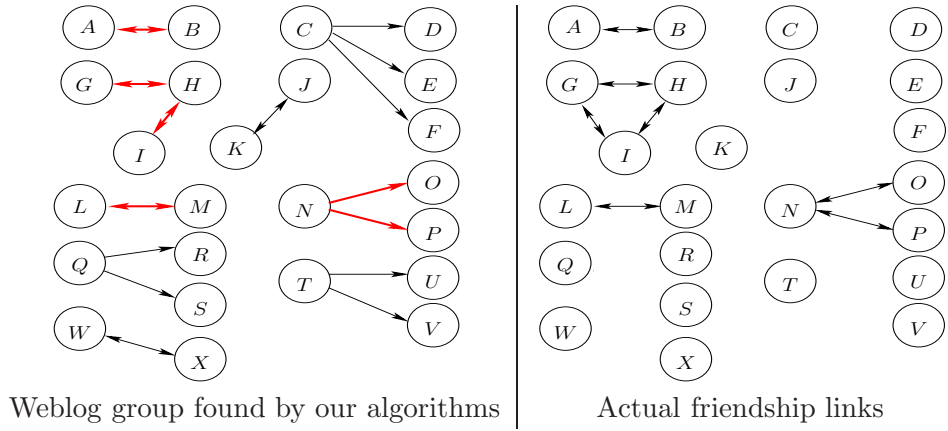

\centering
\begin{tabular}{c | c}
  \resizebox{5.5cm}{5cm}{\input{my_graph.pstex_t}}&
   \resizebox{5.5cm}{5cm}{\input{friendship.pstex_t}}
   \\
Weblog group found by our algorithms
& Actual friendship links
\end{tabular}
\caption{Validation of Weblog group communicational structure on the left against actual 
friendship links on the right.}
\label{fig:weblog}
\end{figure}
constructed from 
friendship lists of people who participated in the conversations 
during that period. 
Fig.~\ref{fig:weblog} shows one of the groups found in the Weblog 
data and the 
corresponding friendship links between the people who participated 
in that group. The fraction of friendship links for this group of
24 actors is 2.5\%, again well above the 0.0008\% 
for a randomly chosen group of 24 actors.

\subsection{General Scoring Functions vs. ``Step'' Function Comparison}

Here we would like to present a comparison of a general scoring function 
and a ``step'' function. 
We will compare a given general propagation delay functions 
$G_1$, $G_2$, $G_3$ and $G_4$,  to a 
``best fit'' step function, see Figure \ref{fig:funct2}.

\begin{figure}[ht]
\centering
\begin{tabular}{c c}
   \resizebox{7cm}{4cm} {\includegraphics[scale=1]{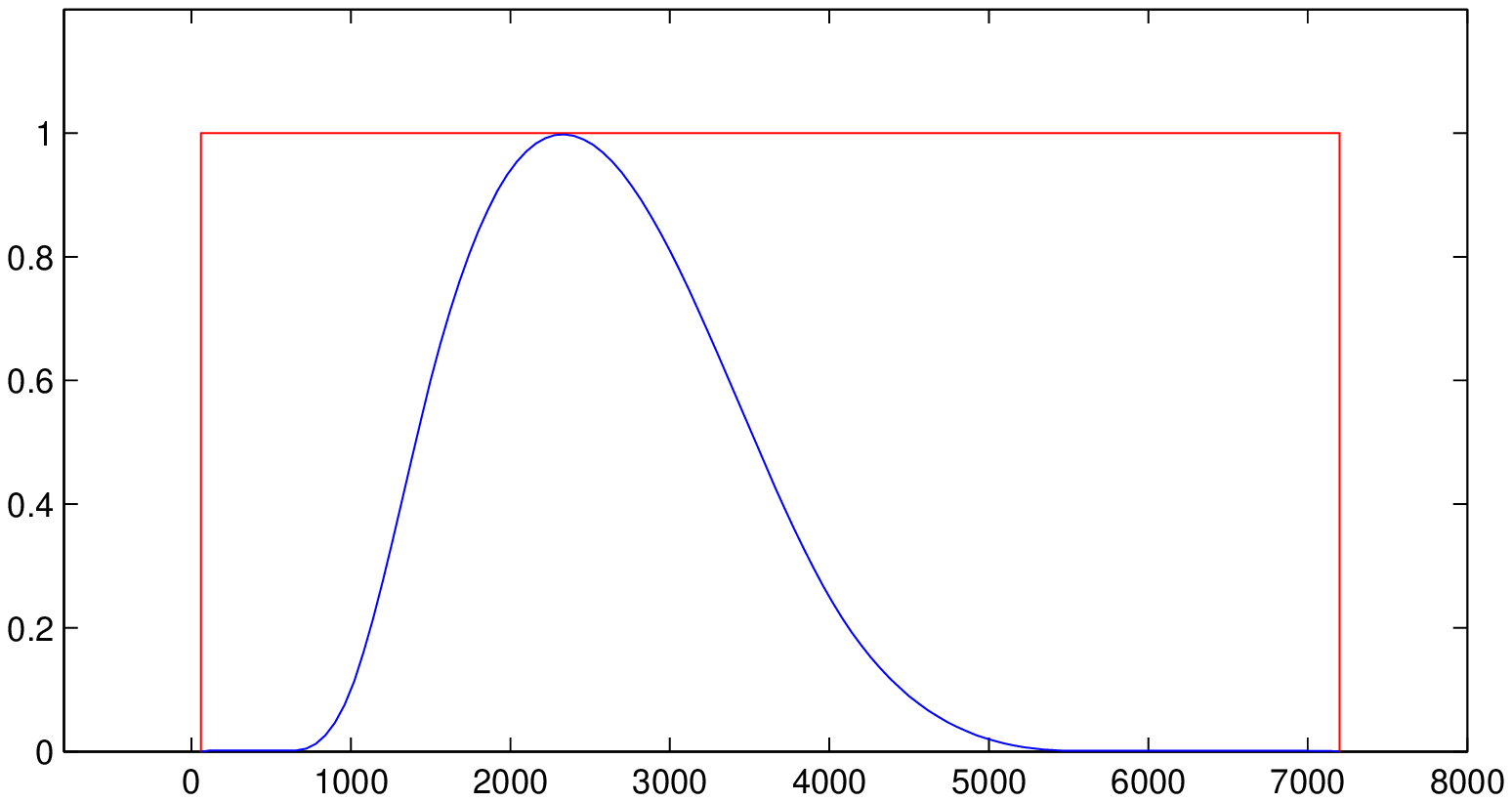}}&
   \resizebox{7cm}{4cm} {\includegraphics[scale=1]{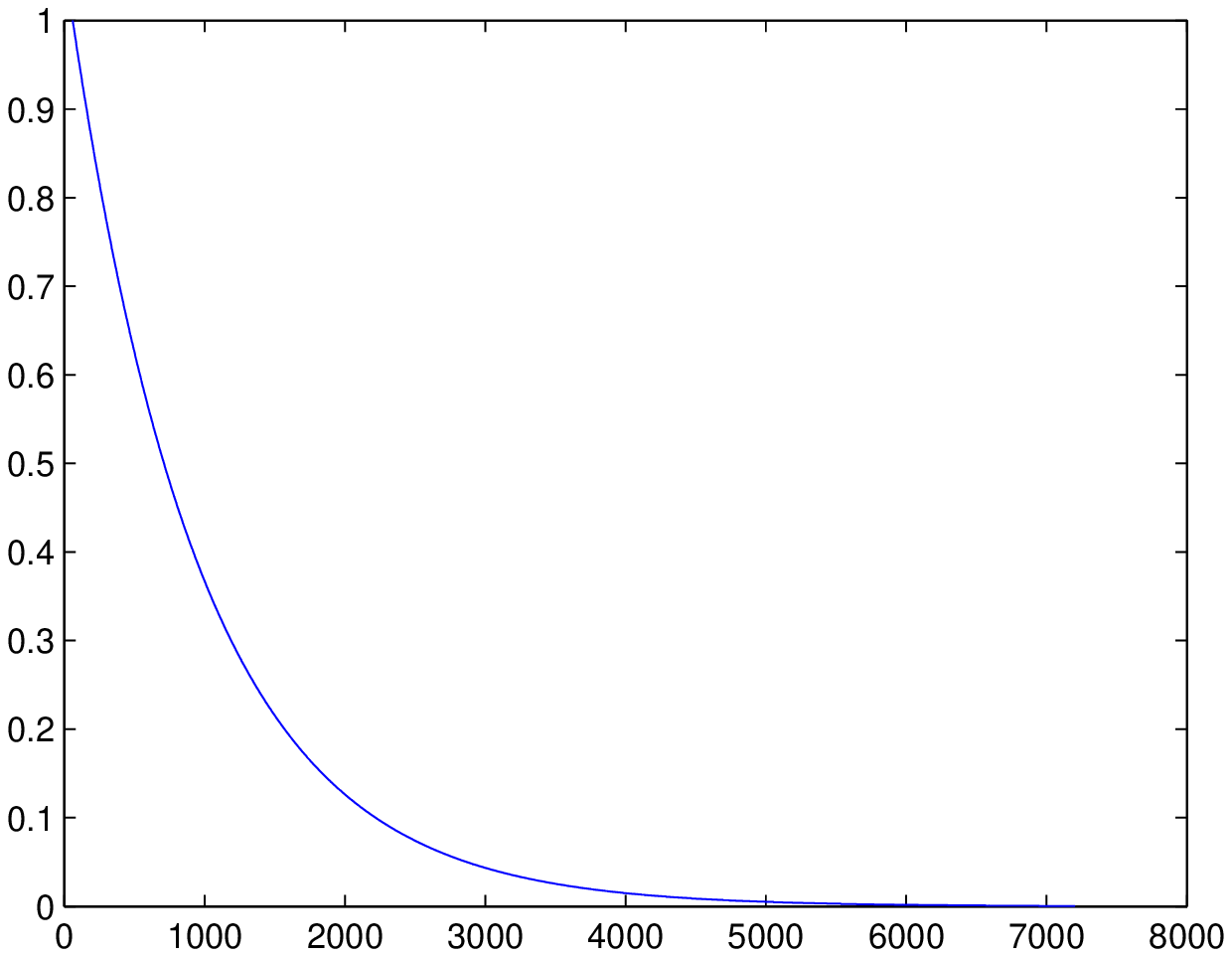}}
   \end{tabular}
  \caption{Step function $H$ and a General Response Function $G_1$ for 
$2D$ Matching on the left and Exponential Decay Response Function $G_4$ on the right.}
  \label{fig:funct2}
\end{figure}


Functions $G_2$ and $G_3$ are respectively linear monotonically increasing and linear 
monotonically decreasing functions, while $G_4$ is generated using a well 
known exponential distribution of 
the form $(y = \lambda \cdot e^{- \lambda \cdot x} )$. 
We generated $G_1$ using a cubic splines interpolation.

For the purposes of this experiment we used an Enron dataset, where we 
looked at the data which represents approximately one year of 
Enron communications and consists of $753,000$ messages. 
We obtained a set of triples of $H$ for the step function 
and a set of triples of $G_1$, $G_2$, $G_3$ and $G_4$ for the general 
propagation functions with causality 
constraint and $G_1'$, $G_2'$, $G_3'$ and $G_4'$ without causality constraint. 
Next we used our distance measure algorithms to measure the 
relative distance between these graphs. Figure \ref{fig:distance1} 
shows the discovered relative distances. 

The results indicate that functions $H$, $G_3$ and $G_4$ produce very similar 
sets of triples, which is explained by the fact that the most of 
the captured triples occur ``early'' and therefore are discovered by 
these somewhat similar functions. Also we can notice that $G_1$ 
and $G_2$ find different sets of triples, while $G_1$ still has a significant 
overlap with $H$, 
we can explain this behavior by the fact that $G_1$ and $G_2$ have 
peaks in different time intervals and thus capture triples occurring 
in those intervals.

\begin{figure}[t]
\begin{center}
  \begin{tabular}{ | l || c | c | c | c | c |}
    \hline
      & $H$ & $G_1$ & $G_2$ & $G_3$ & $G_4$ \\ \hline
    { $G_1$} & 0.63 & - & 0.34 & 0.66 & 0.67 \\ \hline 
    { $G_1'$} & 0.64 & 0.98 & 0.34 & 0.65 & 0.67 \\ \hline 
    { $G_2$} & 0.22 & 0.34 & - & 0.33 & 0.18 \\ \hline 
    { $G_2'$} & 0.23 & 0.34 & 0.97 & 0.34 & 0.18 \\ \hline 
    { $G_3$} & 0.94 & 0.66 & 0.33 & - & 0.88 \\ \hline 
    { $G_3'$} & 0.95 & 0.67 & 0.34 & 0.96 & 0.89 \\ \hline 
    { $G_4$} & 0.90 & 0.67 & 0.18 & 0.88 & - \\ \hline 
    { $G_4'$} & 0.92 & 0.67 & 0.19 & 0.89 & 0.97 \\ \hline 
  \end{tabular}
\end{center}
\caption{Relative similarity between the groups of $H$, $G$s and $G'$s.}
\label{fig:distance1}
\end{figure}

In the current setting we showed that functions 
with peaks at different points will discover different 
triples. Most of the times 
in real data there seems to be no practical need for this added generality, 
however having this ability at hand may prove useful in certain settings. 
Also, since the difference is small compared to the rate of group change in the 
Enron data, hence there is not much value added by a general propagation 
delay function to justify the increase in computation cost 
from linear to quadratic time.

\subsection{Tracking the Evolution of Hidden Groups}

For chains the significance threshold frequencies  were 
$\kappa_{chain}$ = 30 and $\kappa_{sibling}$ = 160. We 
used a sliding window of one year to obtain evolving
hidden groups. 
On each window we obtained the significant 
chains and siblings (frequency $>\kappa$) 
and the 
clusters in the 
corresponding weighted overlap graph. We use 
the clusters to build the 
communication structures and show the evolution of one of the 
hidden groups in Fig.~\ref{fig:enron1} without relying 
on any semantic message 
information. The key person in this hidden group is 
actor $C$, who is {\sf Peggy Heeg}, Senior Vice President of 
El Paso Corporation. 
El Paso Corporation was often partnered with ENRON and was accused of 
raising prices to a record high during the ``blackout'' period in 
California~\cite{internet_ref2,internet_ref1}.

\subsection{Estimating the Rate of Change for Coalitions in the Blogosphere}

Next we would like to show how the approaches of distance measure, 
presented in this thesis, can be used to track the 
evolution and estimate the rate of change of the 
clusterings and groups over time. 
As our example we studied the social network 
of the Blogosphere (Live Journal). 
We found four clusterings $C_1$, $C_2$, $C_3$ and $C_4$ 
by analyzing the same social network at different times. 
Each consecutive clustering 
was constructed one week later than the previous. The task of this 
experiment is to find the amount of change that happened in 
this social network over the period of four weeks. 
The sizes of the clusterings $C_1$, $C_2$, $C_3$ and $C_4$ are 
$81348$, $82056$, $82132$ and $80217$ respectively, 
while the average densities are 
$0.630$, $0.643$, $0.621$ and $0.648$. 

\begin{figure}[b]
\centering
  \begin{tabular}{ | l || c | c | c | c |}
    \hline
      & $C_1$ - $C_2$ & $C_2$ - $C_3$ & $C_3$ - $C_4$ & Average Change\\ \hline
    {\it Best Match} & 4.31 & 5.01 & 4.83 & 4.72 \\ \hline
  \end{tabular}
\caption{The rate of change of the clusterings in Blogosphere over the period of four weeks.}
\label{fig:blogs}
\end{figure}

We can see in the Fig.~\ref{fig:blogs} that the {\it Best Match} and the 
{\it K-center} algorithms imply that the rate of change of groups 
in the blogosphere is relatively high and the groups change very dynamically 
from one week to another. 

\begin{figure*}[ht]
\centering
  \begin{tabular}{ | l || c | c | c | c |}
    \hline
      & $C_1'$ - $C_2'$ & $C_2'$ - $C_3'$ & $C_3'$ - $C_4'$ & Avg. Change\\ \hline
    {\it Best Match} & 0.3 & 0.23 & 0.24 & 0.26 \\ \hline
  \end{tabular}
\caption{The rate of change of the clusterings in the Enron organizational structure from 2000 - 2002.}
\label{fig:enron}
\end{figure*}

\subsection{Estimating the Rate of Change for Groups in the Enron Organizational Structure}

Another experiment we conducted, using the proposed distance measures, 
is estimating the rate of change for groups in the Enron organizational structure. 
We used Enron email corpus and the approach proposed in \cite{hidden2006} 
to obtain clusterings with statistically significant persistent groups in 
several different time intervals.  

On each window we first obtained the significant 
chains and siblings and then the clusterings in the 
corresponding weighted overlap graph. The  
clusterings $C_1'$, $C_2'$, $C_3'$ and $C_4'$ 
with average densities $0.65$, $0.7$, $0.71$ and $0.67$ respectively, 
were computed based on the intervals Sept. 1999 - Sept. 2000,  Mar. 2000 - Mar. 2001, 
Sept. 2000 - Sept. 2001 and Mar. 2001 - Mar. 2002.

Next we used the {\it Best Match}  and {\it K-center} algorithm to track the 
rate of change in the network.  
Fig.~\ref{fig:enron} illustrates the rate of change as well as the average rate 
of change of the clusterings. Notice that the rate of change in the email 
networks over a 6 month period are significantly lower than 
the rate of change in Blogs over a 1 week period. 
Blogs are a significantly more dynamic social network which should be no surprise.

The Fig.~\ref{fig:enron1} illustrates the structure of a single 
group in each of the clusterings as well as it gives a sense of its evolution
from one time interval to next. 

The ability to account for the overlap 
and evolution dynamics is the underlying reason why the distance found   
by the {\it Best Match} and {\it K-center} algorithms is relatively low 
for groups in the ENRON dataset. 

As a conclusion we would like to point out that Blogs and ENRON 
are two completely different social networks. ENRON represents a 
company network, which has the underlying hierarchy of command, 
which is unlikely to change quickly over time, while Blogosphere is a much 
more dynamic social network, where groups and their memberships can change 
rapidly. This behavior is well reflected in the experiments described above.

\subsection{Tree Mining Validation}

Additionally for the purpose of validation, 
we used the tree mining approach in 
conjunction with the heuristic algorithms , in order to 
verify that a discovered tree-like structure actually occurs frequently in 
the data. 
\begin{figure*}[ht]
\centering
  \begin{tabular}{ | l || c | c | c | c |}
    \hline
      & $C_1$ - $T_1$ & $C_2$ - $T_2$ & $C_3$ - $T_3$ & $C_4$ - $T_4$ \\ \hline
    {\it Best Match} & 0.323 & 0.321 & 0.294 & 0.389 \\ \hline
  \end{tabular}
\caption{The similarity between the trees and the clusterings in the Enron organizational structure from 2000 - 2002.}
\label{fig:enron2}
\end{figure*}
For the experiment, we once again used the ENRON email corpus and the time intervals 
Sept. 1999 - Sept. 2000,  Mar. 2000 - Mar. 2001, 
Sept. 2000 - Sept. 2001 and Mar. 2001 - Mar. 2002. For each interval were found 
significant 
chains and siblings and performed the clustering on the weighted graph of overlapping 
triples. The clusterings $C_1$, $C_2$, $C_3$ and $C_4$ were found. 
Next we performed tree mining in order to extract exact tree like communication 
patterns for the same intervals and obtained $T_1$, $T_2$, $T_3$ and $T_4$. 
The same significance threshold frequencies were used 
$\kappa_{chain}$ = 35 and $\kappa_{sibling}$ = 160 when we found 
$C_1$, $C_2$, $C_3$ and $C_4$.

\begin{figure*}[ht]
\centering
  \begin{tabular}{ | l || c | c | c | c |}
    \hline
      & $C'_1$ - $T'_1$ & $C'_2$ - $T'_2$ & $C'_3$ - $T'_3$ & $C'_4$ - $T'_4$ \\ \hline
    {\it Best Match} & 0.411 & 0.407 & 0.414 & 0.41 \\ \hline
  \end{tabular}
\caption{The similarity between the trees and the clusterings in the Blogosphere 
over the period of $4$ weeks}
\label{fig:weblog2}
\end{figure*}

We also performed the same experiment on the Blogosphere, where we randomly 
picked the set of $4$ consecutive weeks and discovered groups by performing 
our heuristic clustering approach to obtain clustering 
$C'_1$, $C'_2$, $C'_3$ and $C'_4$. Next we found exact tree like communication 
structures $T'_1$, $T'_2$, $T'_3$ and $T'_4$ for each week respectively.

We used the {\it Best Match}  and the {\it K-center} algorithms to measure 
the amount of similarity between these two sets. You can find the results of these measurements 
in the Fig.~\ref{fig:enron2} and \ref{fig:weblog2}. The groups which we find 
using a heuristic clustering approach compare well to the actual tree-like 
structures present in the data. 

Additionally we would like to bring your attention to the Fig.~\ref{fig:blogs} and 
\ref{fig:weblog2} to point out that despite the rapid and dynamic rate of change in the 
Blogosphere as a system, the relative distance which is found between respective 
$T$'s and $C$'s remained low. This suggests that our algorithms for discovering planning 
hidden groups are able to perform well for very dynamic systems as Blogosphere as well as 
the more stable systems as ENRON.
Notice that the slightly higher similarity in the ENRON data could be caused by the fact 
that the underlying hierarchy like structure of the ENRON company resembles 
the tree like patterns much more then a chaotic Blogosphere. Nevertheless the 
discovered similarity for the groups in the Blogosphere data is still suggesting that the groups 
we discover using our heuristic approach are similar in their nature to the 
groups discovered by performing tree mining. Thus this section provides yet 
another prove of that our algorithms find real and 
meaningful groups in the streaming communication data by using no message content.

\section{Conclusions}
\label{sec:conclusion}

In this work, we described algorithms for discovering
hidden groups based only on communication data. 
The structure imposed by the need to plan was a very general one, namely
connectivity. Connectivity should be a minimum requirement
for the planning to take place, and perhaps adding further constraints can
increase the accuracy or the efficiency.

In our
algorithms there is no fixed communication cycle and the group's planning 
waves of communications may overlap. 
The algorithm first finds statistically significant 
chain and sibling triples.
Using a heuristic to build from triples, 
we find hidden groups of larger sizes. Using a moving window and matching algorithms
we can track the evolution of the organizational structure as well 
as hidden group membership. Using a tree querying algorithm one can query a 
hierarchical structure to check if it exists in the data. The tree mining algorithm 
finds exactly all of the frequent trees and can be used for verification purposes. 
Our statistical algorithms serve to narrow down the set of possible 
hidden groups that need to be analyzed further.

We validated our algorithms on real data and 
our results indicate that the hidden group algorithms do indeed find
meaningful groups. 
Our algorithms don't use communication content and don't 
differentiate between the natures of the hidden groups discovered, 
for example some of the hidden groups may be 
malicious and some may not. The groups found 
by our algorithms can be further
studied by taking into account the form and the content 
of each communication, to get a better overall 
result and to identify the truly 
suspicious groups. 

\bibliographystyle{abbrv}
\bibliography{IEEEabrv,master}

\end{document}